\newif\iflncs\lncsfalse
\newif\iflong\longtrue
\newif\ifthm\thmtrue
\documentclass{article}

\usepackage[disable]{todonotes}

\usepackage{arxiv}

\usepackage[utf8]{inputenc}
\usepackage[T1]{fontenc}
\usepackage{amsmath,amssymb}
\usepackage[arrow, matrix, curve]{xy}
\usepackage{xspace}
\usepackage{graphicx}
\usepackage[colorlinks=true,linkcolor=blue,citecolor=blue]{hyperref}
\usepackage{float}
\usepackage{dsfont}
\usepackage{braket}
\usepackage{tikz}
\usepackage{verbatim}
\usepackage{enumitem}
\usepackage{mathrsfs}
\usepackage{algorithm}
\usepackage{apxproof}
\usepackage{subcaption}
\usepackage{cleveref}
\usepackage{tabularx}
\usepackage{mathtools}
\usepackage{nicefrac}

\usepackage[noend]{algpseudocode}

\newtheorem{redrule}{Rule}
\newcommand{\Oh}{\ensuremath{\mathcal{O}}}

\newtheorem{myclaim}{Claim}{\itshape}{\rmfamily}
\newtheorem{fact}{Fact}{\bfseries}{\itshape}
\newenvironment{claimproof}{{\noindent\textit{Proof. }}}{\hfill$\blacksquare$\vspace{0.1cm}}
\theoremstyle{plain}
\newtheorem{lemma}{Lemma}
\theoremstyle{plain}
\newtheorem{corollary}{Corollary}
\theoremstyle{plain}
\newtheorem{definition}{Definition}
\theoremstyle{plain}
\newtheorem{theorem}{Theorem}

\newcommand{\lncsqed}{\iflncs\hfill$\qed$\fi}
\DeclareMathOperator{\lp}{lp}

\newcommand{\todom}[2][]{\todo[#1,color=green!50]{#2}}

\newcommand{\todok}[1]{\todo[backgroundcolor=red!60]{ #1}}

\definecolor{myred}{rgb}{1,0.25,0.25}

\newcommand{\prob}[3]{\begin{quote}  \textsc{#1}\\  \textbf{Input:} #2\\  \textbf{Question:} #3\end{quote}}

\newcommand{\W}[1]{\ensuremath{\mathrm{W}[#1]}\xspace}
\newcommand\NP{\ensuremath{\mathrm{NP}}\xspace}

\newcommand\coNP{\ensuremath{\mathrm{coNP}}\xspace}
\newcommand\coNPpoly{\ensuremath{\mathrm{\coNP/poly}}\xspace}
\newcommand\FPT{\ensuremath{\mathrm{FPT}}\xspace}

\newcommand{\VC}{\textsc{Vertex Cover}\xspace}

\newcommand{\IS}{\textsc{Independent Set}\xspace}
\newcommand{\BIC}{\textsc{Biclique}\xspace}
\newcommand{\BIN}{\textsc{Bin Packing}\xspace}

\newcommand{\MCI}{\textsc{MCP}\xspace}
\newcommand{\WMCI}{\textsc{WMCP}\xspace}
\newcommand{\ZMCI}{\textsc{ZWMCP}\xspace}

\newcommand{\MCIlong}{\textsc{Min~$(s,t)$-Cut Prevention}\xspace}
\newcommand{\WMCIlong}{\textsc{Weighted Min~$(s,t)$-Cut Prevention}\xspace}
\newcommand{\ZMCIlong}{\textsc{Zero-Weight Min~$(s,t)$-Cut Prevention}\xspace}

\iflong

\else

\renewcommand{\ZMCIlong}{\textsc{Weighted Min~$(s,t)$-Cut Prevention}\xspace}

\renewcommand{\ZMCI}{\textsc{WMCP}\xspace}
\fi

\newcommand{\KNAP}{\textsc{Knapsack}\xspace}

\newcommand{\vc}{\ensuremath{\mathrm{vc}}}
\newcommand{\pw}{\ensuremath{\mathrm{pw}}}

\newcommand{\fvs}{\ensuremath{\mathrm{fvs}}}

\newcommand{\td}{\ensuremath{\mathrm{td}}}
\newcommand{\tw}{\ensuremath{\mathrm{tw}}}
\newcommand{\mT}{\ensuremath{\mathcal{T}}}

\title{Preventing Small $\mathbf{(s,t)}$-Cuts by Protecting Edges}
\begin{document}

\author{Niels Grüttemeier, Christian Komusiewicz, Nils Morawietz\thanks{Supported by the \iflong Deutsche Forschungsgemeinschaft (DFG)\fi, project OPERAH, KO~3669/5-1.}, and Frank~Sommer
\\
Fachbereich Mathematik und Informatik, Philipps-Universität Marburg, Marburg, Germany \\ \{niegru,komusiewicz,morawietz,fsommer\}@informatik.uni-marburg.de}


\maketitle

\begin{abstract}
  We introduce and study \textsc{Weighted Min~$(s,t)$-Cut Prevention}, where we are
  given a graph~$G=(V,E)$ with vertices $s$~and~$t$ and an edge cost function and the aim
  is to choose an edge set~$D$ of total cost at most~$d$ such that~$G$ has no~$(s,t)$-edge cut of capacity at most~$a$ that is disjoint from~$D$. 
  We show that \textsc{Weighted Min $(s,t)$-Cut Prevention} is NP-hard even on subcubcic graphs when all edges have capacity and cost one and provide a comprehensive study of the parameterized complexity of the problem. 
  We show, for example W[1]-hardness with respect to~$d$ and an FPT algorithm for~$a$.
\end{abstract}

\todo[inline]{\textbf{Parameter:~$a -\text{~min cut}$} FPT? \W1? (for ZMCP \W1-hard (see Biclique reduction))}

\todom[inline]{cut for vs cut in

disjoint from~$D$ -> that avoids~$D$

überall 'solution' nutzen

solution in~$D$ und cut in~$A$ umbenennen

Achtung: solution hat noch keine Kostenschranke für den Verteidiger drin}

\todom[inline]{was geht ab auf planaren Graphen erwähnen?}

\todom[inline]{check strange spacing in statements at beginning}

\todom[inline]{cut -> $(s,t)$-cut}

\todom[inline]{bigger than vs. larger than}

\todo[inline]{Achtung: solution hat noch keine Kostenschranke für den Verteidiger drin}
\todo[inline]{check $P \subseteq S \to E(P) \subseteq S$}

\todom[inline]{cite show-master thesis :)}

\section{Introduction}

Network interdiction is a large class of optimization problems with direct applications in
operations research~\cite{CZ17,CMW98,SPG13,Wood93,Zenk10}. In these problems one player wants to achieve a certain goal (for example
finding a short path between two given vertices $s$~and~$t$), and another player wants to modify the network to prevent this.
Given the enormous importance of the maximum-flow/min-cut problem it comes as no
surprise that two-player games where an attacker wants to decrease the maximum $(s,t)$-flow of a network by deleting edges have been considered~\cite{CZ17,Wood93}. We
study an inverse problem: an attacker wants to find an $(s,t)$-cut of capacity at most~$a$
and a defender wants to protect edges in order to increase the capacity of any
minimum $(s,t)$-cut in~$G$ to at least~$a+1$.  \iflong Alternatively, we may think that the
defender increases the capacity of some edges to~$a+1$ in such a way that the maximum~$(s,t)$-flow of
the resulting network exceeds the given threshold~$a$.\fi The formal problem definition reads as follows.
\prob{\WMCIlong (\WMCI)}{A graph~$G=(V,E)$, two vertices~$s,t\in V$, a cost function~$c: E \to \mathds{N}$, a capacity function~$\omega: E \to \mathds{N}$, and integers~$d$ and~$a$.}{Is there a set~$D \subseteq E$ with~$c(D):=\sum_{e \in D}c(e) \leq d$ such that for every~$(s,t)$-cut~$A\subseteq (E\setminus D)$ in~$G$ we have~$\omega(A) := \sum_{e\in A}\omega(e) > a$?}
 The special case where we have only unit capacities and unit costs is referred to as \MCIlong{} (\MCI).
A different problem also called \textsc{Minimum $st$-Cut Interdiction} has been studied recently~\cite{AAT20} but in this problem the graph is directed and the interdictor may freely choose the amount of increase in edge capacities. In our formulation, the interdictor may only decide to fully protect an edge or to leave it unprocted. To the best of our knowledge, this formulation of
\WMCI{} has not been considered so far. We study the classical complexity of
\WMCI{} and its parameterized complexity with respect to~$a$,~$d$, and important structural parameterizations of the input graph~$G$.

\paragraph{Related Work.}
Many interdiction problems have been studied from a (parameterized) complexity perspective: In \textsc{Matching Interdiction}~\cite{Zenk10}, one wants to remove vertices or edges to decrease the weight of a maximum-weight matching. In the \textsc{Most Vital Edges in MST} problem, one aims to remove edges to decrease the weight of any maximum spanning tree. In \textsc{Shortest-Path Interdiction}~\cite{IW02}, also known as \textsc{Shortest Path Most Vital Edges}~\cite{BFN+19,FHN+18} and \textsc{Minimum Length-Bounded Cut}~\cite{BEH+10}, one wants to remove edges to increase the length of a shortest~$(s,t)$-path above a certain threshold. All of these problems are NP-hard and the study of their classical and parameterized complexity has received a lot of attention~\cite{BFN+19,GS14,FHN+18,Zenk10}.

\paragraph{Our Results.} \begin{table}[t]
  \caption{Parameter overview for \WMCI and \MCI.
    We write \NP-h if the problem is \NP-hard even if the corresponding parameter is a constant.}

\begin{tabularx}{\textwidth}{p{1.4cm} p{1.3cm} p{1.6cm} p{1.6cm} p{1.6cm} p{1.9cm} p{0.1cm} p{1.9cm}}
  \hline
  & $a$ & $d$ & $\Delta$ & $d+\Delta$ & $\vc$ & &$\pw + \fvs$ 
			\\
			\hline
			\WMCI~ & $\FPT$ & \W1-h & \NP-h & \W1-h if~$\Delta = 3$ & weak\-ly~{\NP-h} Thm.~\ref{theo:wmci paraNPhard vc}& & weak\-ly {\NP-h} Thm.~\ref{theo:wmci paraNPhard vc}		
			\\
			 &  &  &  &  & \W1-h  & 	&	 \W1-h 
			\\
			
\ifthm	& Thm.~\ref{theo: FPT a} &  Lem.~\ref{lem: np hard}   &  Thm.~\ref{theo: parametr d plus delta}  & Thm.~\ref{theo: parametr d plus delta} &   Thm.~\ref{thm:w hard vc poly weight}  & & Thm.~\ref{thm:w hard vc poly weight} \\ \fi

			\hline
			\MCI & $\FPT$ & \W1-h & \NP-h & $\FPT$ & $\FPT$ & & \W1-h 	\\

\ifthm 
&	 Thm.~\ref{theo: FPT a} & Lem.~\ref{lem: np hard}   &  Thm.~\ref{theo: parameter degree}  &  Thm.~\ref{theo: parameter degree} &   Thm.~\ref{thm-mci-fpt-vc} & & Thm.~\ref{theo: mci w hard pw}\\
\fi			
			
			\hline	
		\end{tabularx}
		\label{tab: parameter overview}
\end{table}
An overview of our results is given in Table~\ref{tab: parameter
  overview}. We show that \WMCI and \MCI are NP-hard even on subcubic graphs. This
motivates a parameterized complexity study with respect to the natural parameters defender
budged~$d$ and attacker budget~$a$ and with respect to structural parameters of the input
graph~$G$. Here, we consider the  structural parameters treewidth
and vertex cover number of~$G$ as well as pathwidth and feedback vertex set number
of~$G$. Our main results are as follows. \MCI{} and \WMCI{} are \W1-hard with respect to
the defender budget~$d$ and FPT with respect to the attacker budget~$a$. \MCI{} and
\WMCI{} are \W1-hard with respect to the combined parameter pathwidth of~$G$ plus
feedback vertex set number of~$G$ and thus also \W1-hard with respect to the treewidth
of~$G$. The hardness for these parameters motivates a study of the vertex cover
number~$\vc(G)$. We show that~\MCI{} is FPT with respect to~$\vc(G)$, whereas~\WMCI{}
is weakly NP-hard even for~$\vc(G)=2$ and \W1-hard with respect
to~$\vc(G)$ even when all capacities and costs are encoded in unary. Finally, we provide a polynomial kernel for~\WMCI{} parameterized
by~$\vc(G)+a$ and complement this result by showing that \MCI{} and \WMCI{} do not admit
polynomial kernels with respect to the large combined parameter~$d+a+\tw(G)+\lp(G)+\Delta(G)$
where~$\lp(G)$ denotes the length of a longest path in~$G$ and~$\Delta(G)$ denotes the
maximum degree.
Overall, our results give a comprehensive complexity overview of \WMCI{} and~\MCI.

\section{Preliminaries}

For integers~$i$ and~$j$ with~$i \leq j$, we define~$[i,j] := \{k \in \mathds{N}\mid i \leq k \leq j\}$.

An (undirected) graph~$G=(V,E)$ consists of a set of vertices~$V$ and a set of edges~$E \subseteq \binom{V}{2}$.
Throughout this work, let~$n:=|V|$ and~$m:=|E|$. 
For vertex sets~$S\subseteq V$ and~$T\subseteq V$ we denote with~$E_G(S,T) := \{\{s,t\}\in E \mid s\in S, t\in T\}$ the edges between~$S$ and~$T$.
Moreover, we define~$E_G(S) := E_G(S,S)$ and~$E_G(v,S) := E_G(\{v\}, S)$ for~$v\in V$.
For a vertex set~$S\subseteq V$ we denote with~$G[S]:=(S, E_G(S))$ the \emph{induced subgraph of~$S$ in~$G$}.
Moreover, for an edge set~$D \subseteq E$ we let~$G-D:=(V, E \setminus D)$ and~$G[D]:=(V,D)$.
For a vertex~$v\in V$, we denote with~$N_G(v):= \{w\in V\mid \{v,w\}\in E\}$ the \emph{open neighborhood} of~$v$ in~$G$.
Analogously, for a vertex set~$S\subseteq V$, we define~$N_G(S) := \bigcup_{v\in S} N_G(S)\setminus S$.
If~$G$ is clear from the context, we may omit the subscript.
A sequence of distinct vertices~$P=(v_0, \dots, v_k)$ is a \emph{path} or~\emph{$(v_0,v_k)$-path} of length~$k$ in~$G$ if~$\{v_{i-1}, v_{i}\}\in E(G)$ for all~$i \in [1,k]$. 
\iflong We denote with~$V(P)$ the vertices of~$P$ and with~$E(P)$ the edges of~$P$. \fi
Let~$s$ and~$t$ be distinct vertices of~$V$. 
An edge set~$A \subseteq E$ is an~$(s,t)$ (edge)-cut in~$G$ if there is no~$(s,t)$-path in~$G-A$.
A graph~$G=(V,E)$ is \emph{connected} if there is an~$(a,b)$-path in~$G$ for each pair of distinct vertices~$a,b\in V$.
\iflong Moreover, we call a vertex set~$S$ a \emph{connected component} of~$G$ if~$G[S]$ is connected and if there is no~$S'\supset S$ such that~$G[S']$ is connected.\fi
\todom{cut vs edge cut (besser machen weil wegen WG)}

\iflong
\subsection{Graph parameter}

Let~$G=(V,E)$ be a graph. 
Moreover, we denote with~$\Delta(G) := \max\{|N_G(v)|\mid v\in V\}$ the \emph{maximum degree} of~$G$.

\todom[inline]{define vertex cover number?}

A set~$S\subseteq V$ is a \emph{feedback vertex set} for~$G$ if~$G - S$ is acyclic, that is, if for each pair of distinct vertices~$a,b\in V\setminus S$ there is at most one~$(a,b)$-path in~$G - S$.
The size of the smallest size feedback vertex set for~$G$ is denoted by~$\fvs(G)$.

A path composition~$\mathcal{B}$ for a graph~$G=(V,E)$ is a sequence of bags~$B_1, \dots, B_q$ where~$B_j \subseteq V$ for each~$j\in [1,q]$, such that:
\begin{enumerate}
\item for every vertex~$v\in V$, there is at least one~$i \in [1,q]$ with~$v\in B_i$,
\item for each edge~$e\in E$, there is at least one~$i \in [1,q]$ such that~$e\subseteq B_i$, and
\item if~$v\in B_i \cap B_j$ with~$i \leq j$, then~$v\in B_k$ for each~$k\in [i,j]$.
\end{enumerate}
The~\emph{width of a path decomposition~$\mathcal{B}$} is the size of the largest bag in~$\mathcal{B}$ minus one and the \emph{pathwidth} of a graph~$G$ is the minimal width of any path decomposition of~$G$ which is denoted by~$\pw(G)$.

A \emph{tree decomposition} of a graph~$G=(V,E)$ is a pair~$(\mT,\beta)$ consisting of a directed tree~$\mT=(\mathcal{V}, \mathcal{A}, r)$ with root~$r\in \mathcal{V}$ and a function~$\beta: \mathcal{V} \to 2^V$ such that
\begin{enumerate}
\item for every vertex~$v\in V$, there is at least one~$x\in \mathcal{V}$ with~$v\in \beta(x)$,
\item for each edge~$\{u,v\}\in E$, there is at least one~$x\in X$ such that~$u\in \beta(x)$ and~$v\in \beta(x)$, and
\item for each vertex~$v\in V$, the subgraph~$\mT[\mathcal{V}_v]$ is connected, where~$\mathcal{V}_v := \{x\in \mathcal{V}\mid v\in \beta(x)\}$.
\end{enumerate}
We call~$\beta(x)$ the~\emph{bag} of~$x$.
The~\emph{width of a tree decomposition} is the size of the largest bag minus one and the \emph{treewidth} of a graph~$G$ is the minimal width of any tree decomposition of~$G$ denoted by~$\tw(G)$.

We consider tree decompositions with specific properties.
A node~$x\in \mathcal{V}$ is called:
\begin{enumerate}
\item a \emph{leaf node} if~$x$ has no child nodes in~$\mT$,
\item a \emph{forget node} if~$x$ has exactly one child node~$y$ in~$\mT$ and~$\beta(y) = \beta(x) \cup \{v\}$ for some~$v\in V\setminus \beta(x)$,
\item an \emph{introduce node} if~$x$ has exactly one child node~$y$ in~$\mT$ and~$\beta(y) = \beta(x) \setminus \{v\}$ for some~$v\in V\setminus \beta(y)$, or
\item a \emph{join node} if~$x$ has exactly two child nodes~$y$ and~$z$ in~$\mT$ and~$\beta(x) = \beta(y) = \beta(z)$.
\end{enumerate} 
A tree decomposition~$(\mT=(\mathcal{V}, \mathcal{A}, r),\beta)$ is called~\emph{nice} if every node~$x\in \mathcal{V}$ is either a leaf node, a forget node, an introduce node, or a join node.
\todom{evtl vor die nodearten}

For a node~$x\in \mathcal{V}$, we define with~$V_x$ the union of all bags~$\beta(y)$, where~$y$ is reachable from~$x$ in~$\mT$.
Moreover, we set~$G_x := G[V_x]$ and~$E_x := E_G(V_x)$.

The tree-depth~$\td(G)$ is the smallest height of any directed tree~$T=(V(G), A)$ with the property that for each edge~$\{u,w\}\in E(G)$ either~$u$ is an ancestor of~$w$ in~$T$ or vice versa.
\fi

Two instances~$I$ and~$I'$ of the same decision problem~$L$ are equivalent if~$I$ is a yes-instance of~$L$ if and only if~$I'$ is a yes-instance of~$L$.
A \emph{reduction rule} for a decision problem~$L$ is an algorithm~$A$ that transforms any instance~$I$ of~$L$ into another instance~$A(I)$ of~$L$. 
We call~$A$~\emph{safe}, if for each instance~$I$ of~$L$,~$I$ and~$A(I)$ are equivalent instances of~$L$.
A reduction rule~$A$ is exhaustively applied for an instance~$I$ if~$A(I) = I$. 
\todom[inline]{passt das so?}

\iflong
For details on parameterized complexity, we refer to the standard monograph~\cite{CFK+15}. 
\fi

\iflong

The two further variants of \WMCI{} that we study are defined as follows. 

\prob{\ZMCIlong (\ZMCI)}{A graph~$G=(V,E)$, two vertices~$s,t\in V$, a cost function~$c: E \to \mathds{N}$, a capacity function~$\omega: E \to \mathds{N} \cup \{0\}$, and integers~$d$ and~$a$.}{Is there a set~$D \subseteq E$ with~$c(D):=\sum_{e \in D}c(e) \leq d$ such that for every~$(s,t)$-cut~$A\subseteq (E\setminus D)$ in~$G$ it holds that~$\omega(A) := \sum_{e\in A}\omega(e) > a$?}

\prob{\MCIlong (\MCI)}{A graph~$G=(V,E)$, two vertices~$s,t\in V$, and integers~$d$ and~$a$.}{Is there an edge set~$D \subseteq E$ of size at most~$d$ such that every disjoint~$(s,t)$-cut~$A\subseteq (E\setminus D)$ in~$G$ has size more than~$a$?}

Informally, we search for a cheap set of edges~$S$ such that every disjoint~$(s,t)$-cut~$M$ is expensive.\fi{}

Let~$I=(G=(V,E), s,t,c, \omega, d, a)$ be an instance of any of the above problems (in the case of~\MCI,~$c(e):=\omega(e) := 1$ for all~$e\in E$). 
We call an edge set~$D \subseteq E$ a \emph{solution} of~$I$ if every~$(s,t)$-cut~$A\subseteq E \setminus D$ has capacity at least~$a+1$ according to~$\omega$.
A solution~$D$ of~$I$ is called a \emph{minimum solution} of~$I$, if there is no solution~$D'$ of~$I$ with~$c(D') < c(D)$.

\iflong
\subsection{Basic Observations}
\fi
We assume \iflong without loss of generality \fi that~$G$ is connected and that~$c(e)\le d+1$ and~$\omega(e)\le a+1$ for each edge~$e\in E(G)$, as otherwise we can decrease these weights accordingly.
Furthermore, we can assume that~$d\le c(E)$ \iflong where~$c(E)$ denotes the total sum of edge-costs.
Analogously, we can assume that\fi~$a\le\omega(E)$.

\begin{fact}\label{Lemma: Compute Cut}
Let~$G=(V,E)$ be a graph, let~$\omega:E \rightarrow \mathds{N}$ be a capacity function, and let~$D\subseteq E$. Then, in~$n^{\Oh(1)}$ time we can compute an~$(s,t)$-cut~$A \subseteq E \setminus D$ with~$\omega(A) \leq a$ or report that no such~$(s,t)$-cut exists.
\end{fact}
\iflong
\todom{proof irgendwo hin}
\begin{proof}
We define a capacity function~$\omega': E \rightarrow \mathds{N}$ by~$\omega'(e):= a+1$ if~$e \in D$ and~$\omega'(e):=\omega(e)$ otherwise. We then compute a min~$(s,t)$-cut~$A$ in~$G$ with respect to the new capacity function~$\omega'$ in~$n^{\Oh(1)}$ time. If~$\omega'(A) \leq a$, then we return~$A$. Otherwise, we report that no such~$(s,t)$-cut exists.

Observe that if~$\omega'(A) \leq a$, then~$A \subseteq E \setminus D$, since~$\omega'(e)=a+1$ for every~$e \in D$. Otherwise, if~$\omega'(A)>a$, then every $(s,t)$-cut in~$G$ either contains an edge from~$D$ or has capacity bigger than~$a$. Thus, the algorithm is correct.
\lncsqed\end{proof}
\fi

\iflong
\begin{lemma}\label{lem to unit costs}
Let~$I=(G=(V,E),s,t,c,\omega, d,a)$ be an instance of~\ZMCI and let~$e^*=\{v^*, w^*\}\in E$ with~$c(e^*) > 1$.
Moreover, let~$I' := (G'=(V',E'), s, t, c', \omega',d,a)$ be the instance of~\ZMCI obtained by replacing~$e^*$ by an~$(v^*, w^*)$-path with~$c(e^*)$ edges of cost~$1$ and capacity~$\omega(e^*)$.
Then,~$I$ and~$I'$ are equivalent instances of~\ZMCI and~$I'$ can be computed in $\Oh(c(e^*)\cdot|I|)$~time.
\end{lemma}
\begin{proof}
The running time bound follows immediately by the construction.
It remains to prove the correctness.
Let~$E^* := E' \setminus E$ be the edges of the~$(v^*, w^*)$-path in~$G'$ that replaces the edge~$e^*$.

We show that~$I$ is a yes-instances of~\ZMCI if and only if~$I'$ is a yes-instances of~\ZMCI. 

$(\Rightarrow)$ 
Let~$D$ be a solution of~$I$ of cost at most~$d$.

\textbf{Case 1:~$e^* \in D$.} 
We set~$D' := D \setminus \{e^*\} \cup E^*$. 
Note that~$c'(D') \leq d$. 
We show that~$D'$ is a solution of~$I'$.
Assume towards a contradiction that there is an~$(s,t)$-cut~$A'\subseteq E' \setminus D'$ in~$G'$ with~$\omega(A') \leq a$.
By definition of~$D'$ it follows that~$A' \subseteq E \setminus \{e^*\}$ with~$\omega(A') \leq a$.
Moreover, since we obtained~$G'$ from~$G$ by replacing~$e^*$ with a path consisting of the edges~$E^*$ and~$A'$ is disjoint to~$E^*$, it follows that~$A'$ is an~$(s,t)$-cut in~$G$, a contradiction.

\textbf{Case 2:~$e^* \notin D$.} 
Note that~$D \subseteq E'$ and that~$c'(D) \leq d$. 
We show that~$D$ is a solution of~$I'$.
Assume towards a contradiction that there is an~$(s,t)$-cut~$A'\subseteq E' \setminus D$ in~$G'$ with~$\omega(A') \leq a$.

If~$A' \cap E^*  = \emptyset$, then~$A'$ is also an~$(s,t)$-cut disjoint to~$D$ in~$G$ with~$\omega(A') \leq a$. 
A contradiction.
Otherwise,~$A' \cap E^* \neq \emptyset$.
Hence,~$A := A' \setminus E^* \cup \{e^*\}$ is an~$(s,t)$-cut disjoint to~$D$ in~$G$ with~$\omega(A) \leq a$, a contradiction.

$(\Leftarrow)$
Let~$D'$ be a solution of~$I'$.

\textbf{Case 1:~$E^* \subseteq D'$.} 
We set~$D := D' \setminus E^* \cup \{e^*\}$. 
Note that~$c(D) \leq d$. 
We show that~$D$ is a solution of~$I$.
Assume towards a contradiction that there is an~$(s,t)$-cut~$A\subseteq E \setminus A$ in~$G$ with~$\omega(A) \leq a$.
By definition of~$D$, it holds that~$A\subseteq E' \setminus D'$.
Moreover, since we obtained~$G'$ from~$G$ by replacing~$e^*$ with a path consisting of the edges~$E^*$ and~$A$ is disjoint to~$E^*$, it follows that~$A$ is an~$(s,t)$-cut in~$G'$ with~$\omega'(A) \leq a$, a contradiction.

\textbf{Case 2:~$E^* \not \subseteq D'$.} 
We set~$D := D' \setminus E^*$. 
Note that~$c(D) \leq d$. 
We show that~$D$ is a solution of~$I$.
Assume towards a contradiction that there is an~$(s,t)$-cut~$A\subseteq E \setminus D$ in~$G$ with~$\omega(A) \leq a$.

If~$e^* \notin A$, then~$A$ is also an~$(s,t)$-cut disjoint to~$D'$ in~$G'$ with~$\omega'(A) \leq a$. 
A contradiction.
Otherwise,~$e^*\in A$.
Hence,~$A' := A \setminus \{e^*\} \cup \{e'\}$ for some~$e' \in E^* \setminus D'$ is an~$(s,t)$-cut disjoint to~$D'$ in~$G'$ with~$\omega'(A') \leq a$, a contradiction.
\lncsqed\end{proof}

\begin{lemma}\label{lem to unit weights}
Let~$I=(G=(V,E),s,t,c,\omega, d,a)$ be an instance of~\ZMCI and let~$e^*=\{v^*, w^*\}\in E$ with~$\omega(e^*) > 1$.
Moreover, let~$I' := (G'=(V',E'), s, t, c', \omega',d,a)$ be the instance of~\ZMCI obtained by updating the capacity of~$e^*$ to~$1$ and by adding~$\omega(e^*)-1$ many~$(v^*, w^*)$-paths with two edges of cost~$c(e^*)$ and capacity~$1$ each.
Then,~$I$ and~$I'$ are equivalent instances of~\ZMCI and~$I'$ can be computed in $\Oh(\omega(e^*)\cdot |I|)$~time.
\end{lemma}
\begin{proof}
The running time bound follows immediately by the construction.
It remains to prove the correctness.
By~$V^*:=V'\setminus V$ we denote the vertices and by~$E^* := E' \setminus E$ we denote the edges added to~$G$ to obtain the graph~$G'$.
We prove that~$I$ is a yes-instance of \ZMCI if and only if~$I'$ is a yes-instance for \ZMCI.

$(\Rightarrow)$ Let~$D\subseteq E$ be a solution of~$I$ of cost at most~$d$.

\textbf{Case 1:~$e^*\in D$.} 
We set~$D':=D$.
Clearly,~$c(D')\le d$.
We show that~$D'$ is a solution of~$I'$.
Assume towards a contradiction that there is an~$(s,t)$-cut~$A'\subseteq E' \setminus D'$ in~$G'$ with~$\omega(A') \leq a$.
Recall that each edge in~$E^*$ is on a path between~$v^*$ and~$w^*$.
Since~$e^*\in D'$, we conclude that~$A'\cap E^*=\emptyset$.
Hence,~$A'$ is also an~$(s,t)$-cut of capacity at most~$a$ in~$G$, a contradiction.

\textbf{Case 2:~$e^*\notin D$.} 
Observe that~$D\subseteq E'$ and that~$c'(D)\le d$.
We show that~$D$ is a solution of~$I'$.
Assume towards a contradiction that there is an inclusion-minimal~$(s,t)$-cut~$A'\subseteq E' \setminus D$ in~$G'$ with~$\omega(A') \le a$.

If~$A' \cap E^*  = \emptyset$, then~$A'$ is also an~$(s,t)$-cut disjoint to~$D$ in~$G$ with~$\omega(A') \le a$, a contradiction.
Otherwise,~$A' \cap E^* \neq \emptyset$.
Recall that all edges in~$E^*$ are on paths with two edges between~$v^*$ and~$w^*$.
Thus,~$A'$ is also an~$(v^*,w^*)$-cut in~$G'$. 
Hence,~$\{v^*,w^*\}\in A'$ and for each vertex~$z\in V^*$ at least one adjacent edge is contained in~$A'$. 
Since~$|V^*|=\omega(e^*)-1$, we conclude that~$|A'\cap E^*|\ge\omega(e^*)-1$.
Thus,~$A'\setminus E^*$ is an~$(s,t)$-cut of cost at most~$a$ in~$G$, a contradiction.

$(\Leftarrow)$
Let~$D'$ be a solution of~$I'$ of cost at most~$d$.
By~$P_z$ we denote the path~$(v^*,z,w^*)$ for some vertex~$z\in V^*$.

\textbf{Case 1:~$E(P_z)\subseteq D'$ for some~$z\in V^*$ or~$e^*\in D'$.} 
We set~$D := D' \setminus E^* \cup \{e^*\}$. 
Note that since~$c(e)=c(e^*)$ for each edge~$e\in E^*$ we obtain~$c(D) \le d$. 
We show that~$D'$ is a solution of~$I$.
Assume towards a contradiction that there is an~$(s,t)$-cut~$A\subseteq E \setminus D$ in~$G$ with~$\omega(A) \le a$.
By definition of~$D$, we observe that~$A\subseteq E' \setminus D'$.
Moreover, since we obtained~$G'$ from~$G$ by adding~$\omega(e^*)-1$ paths consisting of the edges~$E^*$, and the fact that~$A$ is disjoint to~$E^*$, we conclude that~$A$ is an~$(s,t)$-cut in~$G'$ with~$\omega'(A) \le a$, a contradiction.

\textbf{Case 2:~$E(P_z)\nsubseteq D'$ for each~$z\in V^*$ and~$e^*\notin D'$.} 
We set~$D := D' \setminus E^*$. 
Note that~$c(D) \le d$. 
We show that~$D$ is a solution of~$I$.
Assume towards a contradiction that there is an~$(s,t)$-cut~$A\subseteq E \setminus D$ in~$G$ with~$\omega(A) \le a$.

If~$e^* \notin A$, then~$A$ is also an~$(s,t)$-cut disjoint to~$D'$ in~$G'$ with~$\omega'(A) \leq a$, a contradiction.
Otherwise,~$e^*\in A$.
Let~$A' := A \setminus \{e^*\} \cup \{\{v^*,z\}\mid z\in V^*\}$. 
Note that for each~$e\in E^*\cup\{e^*\}$ we have~$\omega(e)=1$ and that~$|V^*|=\omega(e^*)-1$.
Hence,~$A'$ is an~$(s,t)$-cut disjoint to~$D'$ in~$G'$ with~$\omega'(A') \le a$, a contradiction.
\lncsqed\end{proof}

Recall that we can assume~$c(e)\le d+1$ and~$\omega(e)\le a+1$ for each edge~$e\in E$. Hence, the subsequent application of Lemmas~\ref{lem to unit costs} and~\ref{lem to unit weights} leads to the following.

\begin{corollary}
\label{cor-weighted-to-unweigthed}
Let~$I=(G=(V,E),s,t,c,\omega, d,a)$ be an instance of~\WMCI.
Then, one can compute in $(n+a+d)^{\Oh(1)}$~time an equivalent instance~$I'=(G',s',t',d,a)$ of~\MCI.
\end{corollary}
\fi

The next definition will be a useful tool in several proofs in this work.

\begin{definition}\label{def:edge contract}
Let~$I=(G=(V,E),s,t,c,\omega,d,a)$ be an instance of~\WMCI, and let~$e=\{u,w\}\in E$.
The \emph{merge of~$u$ and~$w$ in~$I$} is the instance~$I'$ obtained from~$I$ by removing~$u$ and~$w$ from~$G$ and adding a new vertex~$v_{\{u,w\}}$ which is adjacent to~$N(\{u,w\})$.
The cost and capacity for each edge in~$E \cap E'$ are set to the corresponding cost and capacity in~$E$, and for each~$x\in N(\{u,w\})$,
\begin{itemize}
\item $c'(\{v_{\{u,w\}},x\}) = \min\{c(e') \mid e'\in E(x, \{u,w\})\}$, and
\item $\omega'(\{v_{\{u,w\}},x\}) = \sum_{e'\in E(x, \{u,w\})} \omega(e')$. 
\end{itemize}
\end{definition}

\begin{redrule}\label{rr: every edge in minimal cut}
If~$G$ contains an edge~$e^*=\{u^*,w^*\}\in E$ which is not contained in any inclusion-minimal~$(s,t)$-cut of capacity at most~$a$ in~$G$, then merge~$u^*$ and~$w^*$.
\end{redrule}

\begin{lemma}\label{lem: every edge in minimal cut}
Rule~\ref{rr: every edge in minimal cut} is safe and can be applied exhaustively in \iflong polynomial \fi time.
\end{lemma}
\iflong
\begin{proof}
Let~$I = (G=(V,E),s,t,c,\omega, d,a)$ be an instance of \WMCI and let~$I'=(G'=(V',E'),s',t',c',\omega',d,a)$ be the merge of~$u^*$ and~$w^*$ in~$I$.
We show that~$I$ and~$I'$ are equivalent instances of~\WMCI.

$(\Rightarrow)$
Let~$D\subseteq E$ be a solution of~$I$ of cost at most~$d$.

\begin{myclaim}\label{claim:unimportant edge not in solution}
The set~$D^* := D\setminus \{e^*\}$ is a solution of~$I$. 
\end{myclaim}
\begin{claimproof}
Assume towards a contradiction that~$D^*$ is not a solution of~$I$.
Then, there is an inclusion-minimal~$(s,t)$-cut~$A\subseteq E\setminus D^*$ of capacity at most~$a$ in~$G$. 
By the condition of Rule~\ref{rr: every edge in minimal cut}, it holds that~$e^* \not \in A$.
Note that~$A$ avoids~$D^*$. 
This contradicts the fact that~$D^*$ is a solution.  
\end{claimproof}

Due to Claim~\ref{claim:unimportant edge not in solution}  we can assume that~$e^* \not \in D$.
We set~$D':= (D \cap E') \cup \{\{v_{e^*}, x\}\mid \{u^*, x\}\in D \text{ or } \{w^*, x\}\in D\}$.
By definition of~$c'$ it follows that~$D'$ has cost at most~$c(D)$.
Hence, it remains to show that~$D'$ is a solution of~$I'$.

Assume towards a contradiction that~$D'$ is not a solution of~$I'$. 
Then, there is an~$(s',t')$-cut~$A'\subseteq E' \setminus D'$ of capacity at most~$a$ in~$G'$.
We set~$A := (A' \cap E) \cup \{e\in E(x, e^*)\mid \{v_{e^*}, x\} \in A'\}$.
Note that~$A \subseteq E \setminus D$.
By definition of~$\omega'$, we obtain that~$\omega(A) = \omega'(A') \leq a$.
Since~$\{x,v_{e^*}\} \in A'$ if and only if~$E(x, e^*) \subseteq A$, and~$A$ and~$A'$ agree on~$E\cap E'$, we obtain that~$A$ is an~$(s,t)$-cut in~$G$ which contradicts the fact that~$D$ is a solution of~$I$.
Consequently,~$I$ is a yes-instance of~\WMCI.

$(\Leftarrow)$
Let~$D'\subseteq E'$ be a solution of~$I'$ of cost at most~$d$.
We set~$D := (D' \cap E) \cup \{e_x \mid \{v_{e^*}, x\} \in D'\}$, where~$e_x$ is an edge in~$E(x, e^*)$ with minimal cost.
By definition of~$c'$ it follows that~$c(D) \leq c'(D')$.
It remains to show that~$D$ is a solution of~$I$.

\todom{evtl ohne widerspruch beweisen?}
Assume towards a contradiction that~$D$ is not a solution of~$I$.
Then, there is an~$(s,t)$-cut~$A^*\subseteq E \setminus D$ of capacity at most~$a$ in~$G$.
Since~$e^*$ is not contained in any inclusion-minimal~$(s,t)$-cut of capacity at most~$a$, there is an inclusion-minimal~$(s,t)$-cut~$A\subseteq A^* \setminus \{e^*\}$ of capacity at most~$a$ in~$G$.
We set~$A' := (E' \cap A) \cup \{\{v_{e^*}, x\} \in E'\mid E(x, e^*) \subseteq A\}$.
By definition of~$\omega'$, we obtain~$\omega'(A') \leq \omega(A) \leq a$.
Since~$\{x,v_{e^*}\} \in A'$ if and only if~$E(x, e^*) \subseteq A$ and~$A$ and~$A'$ agree on~$E\cap E'$, we obtain that~$A'$ is an~$(s',t')$-cut in~$G$ which contradicts the fact that~$D'$ is a solution of~$I'$.
Consequently,~$I$ is a yes-instance of~\WMCI.

It remains to bound the running time. 
Each application of Rule~\ref{rr: every edge in minimal cut} reduces the number of vertices by one, and each such application can be performed in polynomial time, we obtain that, Rule~\ref{rr: every edge in minimal cut} can be exhaustively applied in polynomial time.
\lncsqed\end{proof}
\fi

\section{NP-hardness and Parameterization by the Defender Budget~$d$}

In this section we prove that \MCI is NP-hard and we analyze parameterization by~$d$ and~$\Delta(G)$.
In particular, we provide a complexity dichotomy for~$\Delta(G)$.

\begin{lemma}\label{lem: np hard}
\WMCI{} is \NP-complete and \W1-hard when parameterized by~$d$ even if~$G$ is bipartite,~$\omega(e) = 1$, and~$c(e) \in\Oh(|G|)$ for all~$e\in E$.
\end{lemma}
\begin{proof}
We describe a parameterized reduction from a variant of \IS which is known to be \W1-hard when parameterized by~$k$~\cite{CFK+15,DF13}. 

\prob{Regular-\IS}{An~$r$-regular graph~$G=(V,E)$ for some integer~$r$ and an integer~$k$.}{Is there an independent set~$S \subseteq V$ of size at least~$k$ in~$G$?}

Let~$I:=(G=(V,E),k)$ be an instance of \textsc{$r$-Regular-\IS}. 
We describe how to construct an instance~$I':=(G'=(V',E'),s,t,c,\omega,d,a)$ of~\WMCI in polynomial time such that~$I$ is a yes-instance of \textsc{Regular-\IS} if and only if~$I'$ is a yes-instance of~\WMCI.
 
We start with an empty graph~$G'$ and add all vertices of~$V$ to~$G'$.
For each vertex~$v\in V$ we also add an additional vertex~$v'$.
Furthermore, for each edge~$e\in E$ we add a vertex~$w_e$, and two new vertices~$s$ and~$t$ to~$G'$.
Moreover, we add the edges~$\{s,v\}$,~$\{v,v'\}$ and~$\{v',t\}$ to~$G'$ for each vertex~$v\in V$.
Next, we add the edges~$\{u,w_e\},\{v,w_e\},$ and~$\{w_e,t\}$ to~$G'$ for each edge~$e=\{u,v\}\in E$.
Now, we set~$\omega(e') := 1$ for all~$e'\in E'$.
Furthermore, for each~$e'\in E'$, we set~$c(e') := 1$ if~$s\in e'$ and~$c(e') := k+1$ otherwise.
Finally, we set~$d:= k$ and~$a:=n + kr-1$ where~$n:= |V|$.
This completes the construction of~$I'$.
Observe that~$G'$ is bipartite with one partite set being~$\{t\}\cup V$.
Note that only the edges incident with~$s$ can be protected, since all other edges have cost exactly~$d+1$.

Next, we show that~$I$ is a yes-instance of \textsc{Regular-\IS} if and only if~$I'$ is a yes-instance of~\WMCI.
 
 $(\Rightarrow)$ 
 Let~$S\subseteq V$ be an independent set of~$G$ of size exactly~$k=d$.
 We set~$D':=\{\{s,v\}\mid v\in S\}$.
 Note that~$D'$ has cost exactly~$d$.
 It remains to show that~$D'$ is a solution of~$I'$.
 To this end, we provide~$a+1$ many paths whose edge sets may only intersect in~$D'$.
 
Note that for each vertex~$v\in V\setminus S$ we have a path~$(s,v,v',t)$.
These are~$n-k$ many.
Next, consider a vertex~$v\in S$.
Observe that~$(s,v,v',t)$ and~$\{(s,v,w_e,t) \mid e\in E, v\in e\}$ are~$r+1$ paths only sharing the edge~$\{s,v\}\in D'$.
Since~$|S|=k$ and~$G$ is~$r$-regular, these are~$kr+k$ many paths.
Moreover, since~$S$ is an independent set no two vertices~$u,v\in S$ have a common neighbor~$w_e$ in~$G'$ for~$e=\{u,v\}$.
Hence, there are~$n-k+kr+k=n+kr=a+1$ many~$(s,t)$-paths in~$G'$ whose edge sets only intersect in~$D'$.
 
$(\Leftarrow)$
Suppose that~$I'$ is a yes-instance of~\WMCI.
Let~$D'$ be a solution with cost at most~$d$ of~$I'$.
Recall that~$c(e)=d+1$ for each edge~$e'\in E'$ with~$s\notin e'$.
Hence,~$D'\subseteq \{\{s,v\}\mid v\in V\}$.  
If~$|D'|< d$, then we add exactly~$d-|D'|$ many edges of the form~$\{s,v\}$ which are not already contained in~$D'$ to~$D'$. 
Note that~$D'$ remains a solution of~$I'$.
Thus, in the following we can assume that~$|D'|=d=k$.
Let~$S:=\{v\mid \{s,v\}\in D'\}$. 
We prove that~$S$ is an independent set in~$G$.

Assume towards a contradiction that~$S$ is no independent set in~$G$ and let~$e^*$ be an edge of~$G[S]$.
In the following, we construct an~$(s,t)$-cut~$A\subseteq (E'\setminus D')$ in~$G'$ of size at most~$a$.
Let~$A_{V\setminus S}:=\{\{s,v\}\mid v\notin S\}$,~$A_S:=\{\{v,v'\}\mid v\in S\}$, and~$A_E:=\{\{v, w_e\}\in E'\mid v\in S, e \neq e^*\}$.
We show that~$A:=A_{V\setminus S}\cup A_S\cup A_E \cup \{\{w_{e^*}, t\}\}$ is an~$(s,t)$-cut of size at most~$a$ in~$G'$.
Note that~$|A_{V\setminus S}| + |A_{S}|= n$.
Moreover, since~$|S| = k$ and each vertex~$v\in V$ has degree exactly~$r$ in~$G$,~$|A_E| \le kr-2$.
Hence,~$A$ has capacity at most~$n + kr -1 = a$ since~$\omega(e') = 1$ for each~$e'\in E'$.
It remains to show that~$A$ is an~$(s,t)$-cut in~$G'$. 
Let~$G^* := G' - A$.
Note that~$N_{G^*}(s)=S$ and~$N_{G^*}(v)=\{s\}$ for each~$v\in S \setminus e^*$.
Moreover, note that~$N_{G^*}(v)=\{s,w_{e^*}\}$ for each~$v\in e^*$ and~$N_{G^*}(w_{e^*}) = e^*$.
Hence,~$A$ is an~$(s,t)$-cut in~$G'$ with capacity at most~$a$. A contradiction.

Consequently,~$S$ is an independent set of size~$k$ in~$G$ and, therefore,~$I$ is a yes-instance of \textsc{Regular-\IS}. 
\lncsqed\end{proof}

By applying~\Cref{cor-weighted-to-unweigthed}, we can extend the hardness results to \MCI.
\iflong
Note that if~$k$ is odd,~\Cref{cor-weighted-to-unweigthed} replaces an edge with costs~$k+1$ by a path of even length and thus the resulting instance of \MCI is not bipartite.
Hence, to obtain \W1-hardness in case of odd~$k$, we set~$c(e)=k+2$ for edges not containing~$s$.

\fi

\begin{corollary}\label{Satz:WmciNPh}
\MCI{} is \NP{}-complete and~\W1-hard when parameterized by~$d$, even on bipartite graphs. 
\end{corollary}

Next, we provide a complexity dichotomy for the classical complexity with respect to the maximum degree of the graph.

\iflong
\begin{lemma}\label{lem-poly-max-deg-2} \ZMCI{} can be solved in \iflong polynomial \else $n^{\Oh(1)}$ \fi time on graphs of maximum degree~two.
\end{lemma}
\begin{proof}
Let~$I=(G=(V,E),s,t,c,\omega,d,a)$ be an instance of~\ZMCI where~$G$ has degree at most two.
Recall that we can assume without loss of generality that~$G$ is connected.
Observe that since~$G$ has degree at most two,~$G$ is either a path or a cycle.

\textbf{~$G$ is a path.} 
Let~$P$ be the unique~$(s,t)$-path in~$G$ and let~$E_A := \{e_i \in E(P) \mid \omega(e_i) \leq a\}$ be the set of edges of capacity at most~$a$.
Since~$\{e_i\}$ is an~$(s,t)$-cut of capacity at most~$a$ in~$G$ for every~$e_i\in E_A$, we conclude that~$E_A$ is a subset of every solution of~$I$.
Consequently,~$I$ is a yes-instance of~\ZMCI if and only if~$d \geq c(E_A)$, since every~$(s,t)$-cut~$M \subseteq E \setminus E_A$ has capacity larger than~$a$.

\textbf{~$G$ is a cycle.} 
Let~$P_1$ and~$P_2$ be the unique~$(s,t)$-paths in~$G$.
Moreover, let~$E_A := \{\{e^1_i, e^2_j\}\mid e^1_i \in E(P_1), e^2_j \in E(P_2), \omega(e^1_i) + \omega(e^2_j) \leq a\}$ be the set of minimal~$(s,t)$-cuts of capacity at most~$a$ in~$G$.
Note that every other~$(s,t)$-cut of capacity at most~$a$ is a superset of any~$(s,t)$-cut in~$E_A$.
Hence,~$I$ is a yes-instance of \ZMCI if and only if there is a set~$S\subseteq E(P_1) \cup E(P_2)$ with~$c(S) \leq d$ such that~$S \cap \mathbf{e} \neq \emptyset$ for all~$\mathbf{e}\in E_A$.
This is equivalent to the question if the graph~$G'$ with bipartition~$(E(P_1), E(P_2))$  and edges~$E_A$ has a vertex cover of capacity at most~$d$ with~$c$ as the capacity function.
This can be done in polynomial time.

Consequently, \ZMCI can be solved in polynomial time on graphs of degree at most two.
\lncsqed\end{proof}
\fi
  
\iflong  
\begin{lemma}\label{lem-wcmi-np-h-wrt-d-in-subcubic}\WMCI{} is \NP-hard and \W1-hard when parameterized by~$d$ even on subcubic graphs and even if~$c(e) = 1$ and~$\omega(e) \in\Oh(|G|)$ for all~$e\in E$.
\end{lemma}
\begin{proof}
We reduce from \MCI which is~\W1-hard when parameterized by~$d$ due to Corollary~\ref{Satz:WmciNPh}.
Let~$I=(G=(V,E),s,t,d,a)$ be an instance of \MCI.
Next, we construct an equivalent instance~$I'=(G'=(V',E'),s',t',c',\omega',d',a')$ of \WMCI as follows.

For each vertex~$v\in V$ we add a path~$P_u$ consisting of~$|N(u)|$ vertices to~$G'$.
We denote the vertices of~$P_u$ by~$p_u^1, \ldots p_u^{|N(u)|}$.
In the following, we assume an arbitrary but fixed ordering on~$N(u)$.
Thus, the~$i$-th-vertex of~$N(u)$ is associated with vertex~$p_u^i\in P_u$.
Furthermore, if~$v$ is the~$i$-th neighbor of~$u$ we also write~$p_u^v$ instead of~$p_u^i$ to access neighbor~$v$ more conveniently.
We set~$c'(e)=1$ and~$\omega'(e)=a+1$ for each edge~$e\in E(P_u)$.
Furthermore, for each edge~$\{u,v\}\in E(G)$ we add the edge~$\{p_u^v,p_v^u\}$ to~$G'$ with cost and capacity equal to one.
Next, we set~$s':=p_s^1$ and~$t':=p_t^1$.
Finally, we set~$a':=a$ and~$d':=d$.

Since each vertex in~$P_u$ has exactly one neighbor which is not in~$P_u$, the graph~$G'$ is subcubic.
Next, we prove that~$I$ is a yes-instance of \MCI if and only if~$I'$ is a yes-instance of \WMCI.

Let~$v\in V$ and let~$e$ be an edge of~$P_v$. 
By the fact that~$\omega'(e)=a+1$,~$e$ is not contained in any (inclusion-minimal)~$(s',t')$-cut of capacity at most~$a$.  
Thus, by merging the endpoints of~$e$, we obtain an equivalent instance of~\WMCI due to~\Cref{lem: every edge in minimal cut}.

Let~$I^* =(G^*=(V^*, E^*), s^*,t^*, c^*, \omega^*,d',a')$ be the instance of~\WMCI we obtain after merging the endpoints of all edges contained in any path~$P_v$.
Note that by~\Cref{def:edge contract} it follows that~$G^*$ is isomorphic to~$G$ and~$\omega^*(e) = c^*(e) = 1$ for each~$e\in E^*$.
Thus,~$I$ is a yes-instance of~\MCI if and only if~$I^*$ is a yes-instance of~\WMCI. 
\lncsqed\end{proof}
\fi

\iflong
By \Cref{lem-poly-max-deg-2} and \Cref{lem-wcmi-np-h-wrt-d-in-subcubic} we obtain the following.
\fi

\begin{theorem}\label{theo: parametr d plus delta}
\WMCI can be solved in polynomial time on graphs of maximum degree two.
 \WMCI{} is \NP-hard and \W1-hard when parameterized by~$d$ even on subcubic graphs and even if~$c(e) = 1$ and~$\omega(e) \in\Oh(|G|)$ for all~$e\in E$.
\end{theorem}

\iflong  
Next, we strengthen the NP-hardness of \WMCI on subcubic graphs to \MCI.
\fi

\iflong
\begin{lemma}\label{thm-cmi-np-h-wrt-d-in-subcubic}\MCI{} is \NP{}-complete even on subcubic graph. 
\end{lemma}
\begin{proof}
We reduce from \MCI.
Let~$I=(G=(V,E),s,t,d,a)$ be an instance of \MCI.
We first prove the statement for \WMCI where~$\omega(e)=1$ and~$c(e)\in n^{\Oh(1)}$.
We do this intermediate step to emphasize the main idea of the reduction.
Second, we apply~\Cref{cor-weighted-to-unweigthed} to each edge in the instance of \WMCI to obtain an equivalent instance of \MCI.
Note that since~\Cref{cor-weighted-to-unweigthed} replaces an edge by a path, the resulting instance of \MCI is also subcubic.
Hence, it remains to prove the statement for the restricted version of \WMCI.

Next, we construct an equivalent instance~$I'=(G'=(V',E'),s',t',c',\omega',d',a')$ of \WMCI as follows.
For each vertex~$v\in V$ we add a path~$P_u$ consisting of~$N(u)$ vertices to~$G'$.
We denote the vertices of~$P_u$ by~$p_u^1, \ldots p_u^{|N(u)|}$.
In the following, we assume an arbitrary but fixed ordering on~$N(u)$.
Thus, the~$i$-th-vertex of~$N(u)$ is associated with vertex~$p_u^i\in P_u$.
Furthermore, if~$v$ is the~$i$-th neighbor of~$u$ we also write~$p_u^v$ instead of~$p_u^i$ to access neighbor~$v$ more convenient.
We set~$c'(e):=\omega'(e):=1$ for each edge~$e\in E(P_u)$.
Furthermore, for each edge~$\{u,v\}\in E(G)$ we add the edge~$\{p_u^v,p_v^u\}$ to~$G'$ and set its costs to~$n^2$ and its capacity to one.
Next, we set~$s':=p_s^1$ and~$t':=p_t^1$.
Finally, we set~$a':=a$ and~$d':=dn^2+n(n-1)$. 

Since each vertex in~$P_u$ has exactly one neighbor which is not in~$P_u$, the graph~$G'$ is subcubic.
Next, we prove that~$I$ is a yes-instance of \MCI if and only if~$I'$ is a yes-instance of \WMCI.

$(\Rightarrow)$
Let~$D\subseteq E$ be a solution with cost at most~$d$ of~$I$.
In the following, we construct a solution~$D'\subseteq E'$ with cost at most~$d'$ of~$I'$.

For each edge~$\{u,v\}\in D$ we add the corresponding edge~$\{p_u^v,p_v^u\}$ in~$G'$ to~$D'$.
Since each of these edges has cost~$n^2$, and~$|D|\le s$, these edges contribute at most~$dn^2$ to to cost of~$D'$.
Furthermore, we add each edge in~$P_u$ for each~$u\in V$ to~$D'$. Since~$P_u$ has at most~$n-1$ edges and each edge in~$P_u$ has cost one, all these edges contribute at most~$n(n-1)$ to the total costs.
Hence,~$|D'|\le d'$.
Assume towards a contradiction that~$G'$ has an~$(s',t')$-cut~$A'\subseteq E'\setminus D'$ with~$\omega'(A')\le a'$. 
Since~$E(P_u)\subseteq D'$ for each~$u\in V$, the~$(s,t)$-cut~$A'$ contains only edges of the form~$\{p_u^v,p_v^u\}$ between two different paths.
We define the set~$A$ as the set of corresponding edges of~$A'$ in~$G$.
Since~$|A'|\le a$ we obtain~$|A|\le a$.
Since there is no~$(s,t)$-cut of capacity at most~$a$ in~$G\setminus D$ and~$A\cap D=\emptyset$, we conclude that there exists an~$(s,t)$-path~$(s,w_1,\ldots,w_\ell,t)$ in~$G\setminus A$.
Observe that~$(p_s^1, \ldots, p_s^{w_1}, p_{w_1}^s, \ldots, p_{w_1}^{w_2}, p_{w_2}^{w_1}, \ldots , p_t^{w_\ell},p_t^1)$ is an~$(s',t')$-path in~$G' - A'$, a contradiction to the assumption that~$A'$ is an~$(s',t')$-cut in~$G'$.

$(\Leftarrow)$
Let~$D'\subseteq E'$ be a solution with cost at most~$d'$ of~$I'$.
In the following, we construct a solution~$D\subseteq E$ with cost at most~$d$ of~$I$.

Since~$d'=dn^2+n(n-1)$,~$c'(e)=n^2$ for each edge~$e\notin E(P_u)$ and each~$u\in V$ in~$G'$, and~$c(e)=1$ for each edge~$e\in E(P_u)$ for some~$u\in V$ in~$G'$, we can assume without loss of generality that~$E(P_u)\subseteq D'$ for each~$u\in V$.
We start with an empty set~$D$.
For an edge~$\{p_u^v,p_v^u\}\in D'$ between two different paths, we add the edge~$\{u,v\}$ to~$D$.
Assume towards a contradiction that~$G$ has an~$(s,t)$-cut~$A\subseteq E\setminus D$ with~$\omega(A)\le a$.
We define the set~$A'$ as the set of corresponding edges of~$A$ in~$G'$.
Note that since~$\omega(e)=1$ for each edge~$e\in E'$ we have~$|A'|\le a=a'$.
Since there is no~$(s',t')$-cut of capacity at most~$a$ in~$G'\setminus D'$ and~$A'\cap D'=\emptyset$, we conclude that there exists an~$(s',t')$-path~$(p_s^1, \ldots, p_s^{w_1}, p_{w_1}^s, \ldots, p_{w_1}^{w_2}, p_{w_2}^{w_1}, \ldots , p_t^{w_\ell},p_t^1)$ in~$G'$.
Thus,~$(s, w_1, \ldots , w_\ell, t)$ is an~$(s,t)$-path in~$G\setminus A$, a contradiction to the assumption that~$A$ is an~$(s,t)$-cut in~$G$.
\lncsqed\end{proof}
\fi

\begin{theorem}\label{theo:FPTda}\WMCI can be solved in~$a^d\cdot n^{\Oh(1)}$ time.
\end{theorem}
\iflong
\begin{proof}
Let~$J:=(G,s,t,c,\omega,d,a)$ be an instance of~\WMCI. We prove the theorem by describing a simple search tree algorithm, that we initially call with~$D:=\emptyset$, where~$D$ represents the choice of the defender: \todom{haben wir sonst je von einem defender gesprochen?}

If~$c(D) >d$, then return \emph{no}. 
Otherwise, use the algorithm behind Lemma~\ref{Lemma: Compute Cut} to compute an~$(s,t)$-cut~$A = \{e_1, \dots, e_{z} \} \subseteq E \setminus D$ for some~$z\le |A|$ with~$\omega(A) \leq a$. 
If no such~$(s,t)$-cut exists, then return~\emph{yes}. 
Otherwise, we branch into the cases where~$D:=D \cup \{e_i\}$ for each~$i\in[1, z]$.

The correctness of the algorithm follows from the fact that for every~$(s,t)$-cut~$A \subseteq E \setminus D$ with~$\omega(A) \leq a$, at least one of the edges of~$A$ must be contained in any solution of~$J$. 
It remains to consider the running time of the algorithm. 
We have~$\omega(e) \geq 1$ for every edge~$e$ in any \WMCI instance.
Hence,~$|A| \leq a$ and therefore, the search tree algorithm branches into at most~$a$ cases. Furthermore, after every branching step,~$c(D)$ increases by at least~$1$, since we add one additional edge to~$D$ and we have~$c(e) \geq 1$ for every edge~$e$. Thus, the depth of the search tree is at most~$d$. Together with the running time from Lemma~\ref{Lemma: Compute Cut}, we obtain a total running time of~$a^d n^{\Oh(1)}$.
\lncsqed\end{proof}
\fi

\iflong
\begin{lemma}\label{lem-fpt-delta-and-d}
\MCI can be solved in~$((d/2+1) \cdot \Delta(G))^d \cdot n^{\Oh(1)}$~time\iflong , where~$\Delta(G)$ denotes the maximum degree of the input graph\fi.
\end{lemma}
\begin{proof}
Let~$J:=(G,s,t,d,a)$ be an instance of~\MCI. We prove the theorem by showing that~$a \leq d \cdot \Delta$ in non-trivial instances of~\MCI. Together with Theorem~\ref{theo:FPTda}, we then obtain fixed-parameter tractability for~$d+\Delta$.

If~$G$ contains an~$(s,t)$-path with~at most~$d$ edges, then~$J$ is a trivial~\emph{yes}-instance. Thus, we may assume that for every~$D \subseteq E$ with~$|D| \leq d$, there is no~$(s,t)$-path in~$G$ that contains only edges from~$D$. We use this assumption to prove the following claim.
\begin{myclaim}
\label{claim-a-bounded-by-d-and-delta}
If~$a \geq (d/2+1) \cdot \Delta$, then~$J$ is a~\emph{no}-instance.
\end{myclaim}

\begin{claimproof}
Let~$D \subseteq E$ with~$|D|\leq d$. 
We prove that there exists an~$(s,t)$-cut~$A \subseteq E$ of size at most~$a$. 
To this end, consider the graph~$G_D:=(V,D)$ consisting only of the edges in~$D$. 
Since~$|D| \leq d$ we know that there is no~$(s,t)$-path in~$G_D$. Thus,~$s$ and~$t$ are in distinct connected components~$C_s \subseteq V$ and~$C_t \subseteq V$ in~$G_D$. 
Furthermore, observe that in at least one of the induced graphs~$G_D[C_s]$ or~$G_D[C_t]$, there are at most~$d/2$ edges. 
Without loss of generality, assume that this is the case for~$G_D[C_s]$. 
Then,~$|C_D| \leq d/2+1$. 
We define~$A:=\bigcup_{v \in C_D} X(v)$, where~$X(v) \subseteq E \setminus D$ is the set of all edges in~$E \setminus D$ that are incident with~$v$ in~$G$. 
Note that~$A \subseteq E \setminus D$ and that~$|A| \leq |C_D| \cdot \Delta = (d/2+1) \cdot \Delta \leq a$. 
Moreover,~$A$ is an~$(s,t)$-cut in~$G$ since~$t \not \in C_s$.
%
%
%
\end{claimproof}

By Claim~\ref{claim-a-bounded-by-d-and-delta}, we conclude that for every non-trivial instnace of~\MCI, we have~$a \leq d \cdot \Delta$. Together with Theorem~\ref{theo:FPTda}, we obtain that~\MCI can be solved in~$((d/2+1)\cdot \Delta)^d \cdot n^{\Oh(1)}$~time.
\lncsqed\end{proof}
\fi
  
\iflong  
By \Cref{thm-cmi-np-h-wrt-d-in-subcubic} and \Cref{lem-fpt-delta-and-d} we obtain the following.  
\fi

\begin{theorem}\label{theo: parameter degree}
\MCI is \NP{}-complete even on subcubic graphs. 
Furthermore, \MCI can be solved in~$((d/2+1) \cdot \Delta(G))^d \cdot n^{\Oh(1)}$~time.
\end{theorem}

\section{Parameterization by the Attacker Budget}
In this section, we show that~\WMCI admits an \FPT-algorithm for the parameter~$a$. 
To this end, we first provide an algorithm with a running time of~$a^{f(\tw(G))}\cdot n$ for some computable function~$f$, where~$\tw(G)$ denotes the treewidth of the graph.
Afterwards, we show that for every instance of~\WMCI we can obtain an equivalent instance~$I'$ of \WMCI in polynomial time, where every edge is contained in an inclusion-minimal~$(s,t)$-cut of size at most~$a$ in~$I'$.
Due to previous results~\cite{GJS17,MOR13}, the graph of~$I'$ then has treewidth at most~$g(a)$ for some computable function~$g$.
In combination with the algorithm for~$a$ and~$\tw(G)$, we thus obtain the stated \FPT-algorithm for the parameter~$a$.

The algorithm with a running time of~$a^{f(\tw(G))}\cdot n$ relies on dynamic programming over a tree decomposition.
Essentially, what the attacker can achieve in the current subgraph is to disconnect specific parts of the bag and thus obtain a cheap partition.
Roughly speaking, the algorithm computes the minimum cost for an edge set~$D$ such that each choice of the attacker to obtain any partition disjoint from~$D$ is expensive. 
Hence, before we describe the algorithm, we first introduce some notations for partitions. 

\todom{note somewhere that algos work even if~$d$ is not encoded in unary}

\todom{motivation für partitions/ informelle idee des DPs schon hier}
Let~$X$ be a set.
We denote with~$B(X)$ the collection of all partitions of~$X$. 
Let~$P\in B(X)$ be a partition of~$X$ and let~$v \in X$.
Then, we define with~$P - v := \{R \setminus \{v\}\mid R\in P\} \setminus \{\emptyset\}$ the partition of~$X\setminus\{v\}$ after removing~$v$ from~$P$.
Analogously, for every~$w\not \in X$ we define~$P+w := \{P'\in B(X\cup\{w\})\mid P'-w = P\}$.
Note that~$B(X\setminus \{v\}) = \{P-v\mid P \in B(X)\}$ and~$B(X\cup \{w\}) = \{P+w\mid P\in B(X)\}$.
Moreover, we denote with~$P(v)$ the unique set of~$P$ containing~$v$ for a partition~$P$ of~$X$ and an element~$v\in X$.

Let~$(\mT:=(\mathcal{V}, \mathcal{A}, r),\beta)$ be a tree decomposition of a graph~$G$.
\iflong Recall that for\fi a node~$x\in \mathcal{V}$, we define with~$V_x$ the union of all bags~$\beta(y)$, where~$y$ is reachable from~$x$ in~$\mT$,~$G_x := G[V_x]$, and~$E_x := E_G(V_x)$.

Let~$P$ be a partition of~$\beta(x)$, then we call an edge set~$A\subseteq E_x$ a~\emph{partition-cut for~$P$ in~$G_x$} if~$v$ and~$w$ are in different connected components in~$G_x - A$ for every pair of distinct vertices~$\{v,w\}$ of~$\beta(x)$ with~$P(v) \neq P(w)$.
Note that all edges between distinct sets of~$P$ are contained in every partition-cut for~$P$ in~$G_x$.

\begin{theorem}\label{thm:tw und a}
Let~$\tw(G)$ denote the treewidth of~$G$.
Then, \ZMCI can be solved in~${a^{\tw(G)}}^{\Oh(\tw(G))} \cdot n +m$ time.
\end{theorem}
\begin{proof}
Let~$I=(G=(V,E),s,t,c,\omega,d,a)$ be an instance of~\ZMCI.
In the following, we assume that there is no edge~$\{s,t\}\in E$ since if~$c(\{s,t\}) \leq d$, then~$\{\{s,t\}\}$ is a valid solution with cost at most~$d$ and, thus,~$I$ is a trivial yes-instance of~\ZMCI.
Otherwise, this edge is contained in every~$(s,t)$-cut and, thus, we can simply remove the edge from the graph and reduce~$a$ by~$\omega(\{s,t\})$.

We describe a dynamic programming algorithm on a tree decomposition.
First, we compute a nice tree decomposition~$(\mT=(\mathcal{V}, \mathcal{A}, r), \beta')$ of~$G - \{s,t\}$ with~$|\mathcal{V}| \leq 4n$ such that the bag of the root and the bag of each leaf is the empty set in~$\tw^{\Oh(\tw^3)}\cdot n +m$ time~\cite{K94,B96}. \todom{cite anderes papier/blaues buch für laufzeit von nice tree decomposition finden}
Next, we set~$\beta(x) := \beta'(x) \cup \{s,t\}$ for each~$x\in \mathcal{V}$. 
Note that~$(\mT, \beta)$ is a tree decomposition of width at most~$\tw+2$ for~$G$.
Recall that for a node~$x\in \mathcal{V}$, the vertex set~$V_x$ is the union of all bags~$\beta(y)$, where~$y$ is reachable from~$x$ in~$\mT$,~$G_x := G[V_x]$, and~$E_x := E_G(V_x)$.

The dynamic programming table~$T$ has entries of type~$T[x,f_x, D_x]$ with~$x\in \mathcal{V}$,~$f_x: B(\beta(x)) \to [0, a+1]$, and~$D_x \subseteq E(\beta(x))$. 
 Each entry stores the minimal cost of an edge set~$D \subseteq E_x$ with~$D_x := D \cap E(\beta(x))$ such that for every~$P \in B(\beta(x))$ the capacity of every partition-cut~$A \subseteq E_x\setminus D$ of~$P$ in~$G_x$ is at least~$f_x(P)$.
 
For each entry of~$T$, we will sketch the proof of the correctness of its recurrence. 
The formal correctness proof is then direct and thus omitted.
 
We start to fill the table~$T$ by setting for each leaf node~$\ell$ of~$\mT$:
 $$T[\ell, f_\ell, \emptyset]:= \begin{cases}0 & \text{if~}f_\ell(\{\{s\},\{t\}\}) = f_\ell(\{\{s,t\}\}) = 0, \\
 \infty & \text{otherwise.} \end{cases}.$$

 Recall that~$\beta(\ell) = \{s,t\}$ and that we assumed that there is no edge between~$s$ and~$t$ in~$G$. 
Hence,~$G_\ell$ contains no edges and, thus, the empty set is a partition-cut for both~$\{\{s\},\{t\}\}$ and~$\{\{s,t\}\}$, and has capacity zero. 
\todom{evtl erklären warum~$\infty$ } 

To compute the remaining entries~$T[x,f_x,D_x]$, we distinguish between the three types of non-leaf nodes.

\textbf{Forget node}: Let~$x$ be a forget node with child node~$y$ and let~$v$ be the unique vertex in~$\beta(y) \setminus \beta(x)$. 
Then we compute the table entries for~$x$ by:
 $$T[x, f_x, D_x] := \min_{E_v \subseteq E(v, \beta(x))} T[y, f_y, D_x\cup E_v]$$ where~$f_y(P) := f_x(P - v)$ for each~$P\in B(\beta(y))$.
 
The idea behind the definition of~$f_y(P)$ is that every partition cut for~$P$ in~$G_y$ must be as expensive as the partition cut of the unique partition of~$\beta(x)$ that agrees with~$P$ on~$\beta(x)$.
By the fact that~$G_x = G_y$, it follows that for each partition~$P\in B(\beta(x))$, an edge set~$A\subseteq E_x$ is a partition-cut for~$P$ in~$G_x$ if and only if~$A$ is also a partition-cut for some~$P'\in P+v$ in~$G_x$.
Since we are looking for the minimal costs of an edge set~$D \subseteq E_x$ such that every partition-cut disjoint from~$D$ for~$P -v$ in~$G_x$ has capacity at least~$f_x(P-v)$, it is thus necessary and sufficient that every partition-cut for~$P$ in~$G_y$ has capacity at least~$f_x(P-v)$.

\textbf{Introduce node}: Let~$x$ be an introduce node with child node~$y$ and let~$v$ be the unique vertex in~$\beta(x) \setminus \beta(y)$. 
Then we compute the table entries for~$x$ by:
 $$T[x, f_x, D_x] := T[y, f_y, D_x\cap E(\beta(y))] + c(D_x \setminus E(\beta(y)))$$ where~$f_y(P) := \max(\{0\} \cup \{f_x(P') - \omega(A_{P'})\mid P'\in (P+v), D_x  \cap A_{P'} = \emptyset\})$ for each~$P\in B(\beta(y))$ and~$A_{P'} := E(v, \beta(y) \setminus P'(v))$.

The idea behind the definition of~$f_y(P)$ is that, since every partition in~$P + v$ agrees with~$P$ in~$\beta(y)$, every partition cut for~$P$ in~$G_y$ must be sufficiently large to ensure that every partition cut for any partition in~$P + v$ is as least as expensive as desired.
Since we are looking for the minimum cost of an edge set~$D\subseteq E_x$ which intersects with~$E(\beta(x))$ in exactly the set~$D_x$, the cost of~$D$ is exactly~$c(D \cap E(\beta(y))) + c(D_x \setminus E(\beta(y)))$.
Let~$P'\in B(\beta(x))$.
Note that~$A_{P'}$ is a subset of every partition-cut for~$P'$ in~$G_x$.
Hence, if~$D_x  \cap A_{P'} = \emptyset$, then~$f_y(P'-v)$ has to be at least~$f_x(P') - \omega(A_{P'})$.
Otherwise, if~$D_x  \cap A_{P'} \neq \emptyset$, then there is no partition-cut for~$P'$ in~$G_x$ disjoint from~$D$.

\textbf{Join node}: Let~$x$ be a join node with child nodes~$y$ and~$z$.
Then we compute the table entries for~$x$ by:
 $$T[x, f_x, D_x] := \min_{f_y : B(\beta(y)) \to [0, a+1]} T[y, f_y, D_x] + T[z, f_z, D_x] - c(D_x)$$ where the mapping~$f_z$ is given by~$$f_z(P) := \max\left(0, \min\left(a+1, f_x(P) - f_y(P) + \omega(E(\beta(x)) \setminus E(P))\right)\right)$$ with $E(P) := \cup_{R\in P}E(R)$ for each~$P\in B(\beta(z))$.

The idea behind the definition of~$f_z(P)$ is that the no partition cut for~$P$ in~$G_z$ is more expensive than the sum of any combination of partition cuts for~$P$ in~$G_y$ and~$G_z$ minus the capacity of the cut-edges in the current bag.
Recall that we are looking for the minimum cost of an edge set~$D\subseteq E_x$ such that for each partition~$P\in B(\beta(x))$, every partition-cut for~$P$ in~$G_x$ disjoint from~$D$ has capacity at least~$f_x(P)$.
Since~$E_y \cap E_z = E(\beta(x))$ it follows that the cost of~$D$ is~$c(S_y) + c(S_z) - c(D_x)$, where~$S_y := E_y \cap D$ and~$S_z := E_z \cap D$. 
Moreover, note that for every partition~$P\in B(\beta(x))$, every partition-cut~$A_\alpha \subseteq E_\alpha$ for~$P$ in~$G_\alpha$ has to contain all edges of~$E(\beta(x)) \setminus E(P)$, where~$\alpha\in\{x,y,z\}$.
Thus, we have to guarantee that~$f_y(P) + f_z(P) - \omega(E(\beta(x)) \setminus E(P)) \geq f_x(P)$,~$f_y(P) > f_x(P)$, or~$f_z(P) > f_x(P)$.
 
Then, there is a solution~$D$ of cost at most~$d$ of~$I$ if and only if~$T[r, f_r, \emptyset] \leq d$, where~$r$ is the root of~$\mT$,~$f_r(\{\{s,t\}\}) = 0$ and~$f_r(\{\{s\},\{t\}\}) = a+1$.
Moreover, the corresponding set~$D$ can be found via traceback.
\iflong
It remains to show the running time.

\newcommand{\newTW}{k}

For every node~$x$ of~$\mT$, there are~$(a+2)^{|B(\beta(x))|}\cdot 2^{|\beta(x)|^2}$   entries. 
Since~$(\mT, \beta)$ has at most~$4n$ bags, each bag contains at most~$\newTW := \tw + 3$ vertices, and~$|B(X)| \leq |X|^{|X|}$, the dynamic programming table contains at most~$4n\cdot {(a+2)^{\newTW^{\newTW}}} \cdot 2^{\newTW^2}$ entries.
Now, we bound the running times of the four types of bags.

\begin{itemize}
\item An entry for a leaf node can be computed in~$\Oh(1)$ time.

\item For a forget node, we can compute the function~$f_y$ in~${{\newTW^{\newTW}}}\cdot \newTW^{\Oh(1)}$ time and iterate over all possible choices for~$E_v$ in~$2^k$ time. Thus, an entry in~${{\newTW^{\newTW}}}\cdot 2^{\newTW} \cdot \newTW^{\Oh(1)}$ time.

\item For an introduce node, we can compute the function~$f_y$ in~${{\newTW^{\newTW}}}\cdot \newTW^{\Oh(1)}$ time and thus the entry in the same running time.

\item For a join node, we have~${(a+2)^{\newTW}}^{\newTW}$ possibilities for~$f_y$ and for each of them, we can compute~$f_z$ in~${{\newTW^{\newTW}}}\cdot \newTW^{\Oh(1)}$ time. 
Hence, for a join node, we can compute an entry in~${(a+2)^{\newTW^{\newTW}}}\cdot \newTW^{\newTW} \cdot \newTW^{\Oh(1)}$ time.
\end{itemize}

The join nodes have the worst running time for any entry.
Thus, we can compute all entries of~$T$ in~${(a+2)^{2(\tw+3)^{\tw+3}}} \cdot {(\tw+3)}^{\tw+3} \cdot 2^{(\tw+3)^2}\cdot \tw^{\Oh(1)} \cdot n$ time and obtain the stated running time. 
\fi
\lncsqed\end{proof}

Next, we show that we can use~\Cref{thm:tw und a} to obtain an~\FPT-algorithm for~\WMCI when parameterized by~$a$.
To this end, we first obtain the following corollary which follows from a result of Gutin et al.~\cite[Lemma~12]{GJS17}.

\begin{corollary}\label{cor: bounded tw}
Let~$G=(V,E)$ be a graph, let~$s$ and~$t$ be distinct vertices of~$G$, and let~$a$ be an integer. 
If every edge~$e\in E$ is contained in an inclusion-minimal~$(s,t)$-cut of size at most~$a$, then~$\tw(G) \leq g(a)$ for some computable function~$g$.
\end{corollary}

Hence, to obtain an~\FPT-algorithm for~\WMCI with the parameter~$a$, we only have to find an equivalent instance in polynomial time where each edge is contained in some inclusion-minimal~$(s,t)$-cut of size at most~$a$.
Since each edge in an instance of~\WMCI has capacity at least one, by applying Rule~\ref{rr: every edge in minimal cut} exhaustively we obtain an equivalent instance of~\WMCI where each edge is contained in some inclusion-minimal~$(s,t)$-cut of \emph{size} at most~$a$.
Hence, we obtain the following by combining~\Cref{lem: every edge in minimal cut},~\Cref{cor: bounded tw}, and~\Cref{thm:tw und a}.

\begin{theorem}\label{theo: FPT a}
\WMCI is~\FPT when parameterized by~$a$.
\end{theorem}

\iflong
\todom{explain why not possible for \ZMCI}
Note that this is not possible for~\ZMCI due to~\Cref{thm: w hard d and a}.

\begin{theorem}\label{thm: w hard d and a}
\ZMCI{} is \W1-hard when parameterized by~$d+a$ even if~$\omega(e) \in \{0,1\}$ for all~$e\in E$.
\end{theorem}

\begin{proof}
We describe a parameterized reduction from \BIC which is known to be \W1-hard when parameterized by~$k$~\cite{CFK+15}.

\prob{\BIC}{A bipartite graph~$G=(X \cup Y,E)$  with partite sets~$X$ and~$Y$ and an integer~$k$.}{Does~$G$ contain a $(k,k)$-biclique?}

Let~$I=(G=(X \cup Y,E),k)$ be an instance of~\BIC. 
Now, we describe how to construct an instance~$I'=(G'=(V',E'),s,t,c,\omega,d,a)$ of~\ZMCI in polynomial time such that~$I$ is a yes-instance of~\BIC if and only if~$I'$ is a yes-instance of~\ZMCI.
  
The graph~$G'$ contains the graph~$G$ as a copy together with two new vertices~$s$ and~$t$ and edges~$F:=\{\{s,x\}\mid x\in X\} \cup \{\{y,t\}\mid y\in Y\}$.
Furthermore, each edge~$\{x,y\}\in E$ is subdivided by a new vertex~$w_{xy}$ in~$G'$. 
Hence, the graph~$G'$ contains the edges~$\{x,w_{xy}\}$ and~$\{w_{xy},y\}$ instead of the edge~$\{x,y\}$.
We define~$E_X:=\{\{x,w_{xy}\}\mid x\in X\}$ and~$E_Y:=\{\{w_{xy},y\}\mid y\in Y\}$.
We set~$d:=(2k+1)k^2+2k=2k^3+k^2+2k$ and for each edge~$\{x,y\}\in E$ we set~$c(\{x,w_{xy}\}):=d+1$,~$\omega(\{x,w_{xy}\}):=1$,~$c(\{w_{xy},y\}):=2k+1$, and~$\omega(\{w_{xy},y\}):=0$.
Furthermore, for each edge~$e\in F$ we set~$c(e):=1$, and~$\omega(e):=0$.
Finally, we set~$a := k^2-1$ which completes the construction of~$I'$.

$(\Rightarrow)$ 
Suppose that~$I$ is a yes-instance of \BIC. 
Then there exist sets~$S_X \subseteq X$ and~$S_Y\subseteq Y$ of size~$k$ each, such that~$\{x,y\}\in E$ for all~$x\in S_X$ and~$y\in S_Y$.
We set~$D:=\{\{w_{xy},y\}\mid x\in S_X, y\in S_Y\}\cup \{\{s,x\}\mid x\in S_X\}\cup \{\{y,t\}\mid y\in S_Y\}$.
Observe that~$|D| =(2k+1)k^2+2k = d$. 
Next, we show that~$D$ is a solution of~$I'$.

Consider for each~$x\in S_X$ and for each~$y\in S_Y$ the~$(s,t)$-path~$(s,x,w_{xy},y,t)$. 
Since~$\{s,x\}\in D$,~$\{w_{xy},y\}\in D$, and~$\{y,t\}\in D$, each~$(s,t)$-cut~$A$ has to contain the edge~$\{x,w_{xy}\}$. 
By the fact that~$\omega(\{x,w_{xy}\})=1$ for each~$x\in S_X$ and each~$y\in S_Y$ we observe that~$\omega(A) \geq |S_X\times S_Y| = k^2 = a+1$.
Hence,~$I'$ is a yes-instance of~\ZMCI.


$(\Leftarrow)$
Let~$D$ be a solution with cost at most~$d$ of~$I'$.
Observe that for each edge~$e\in E_X$ we have~$d(e)=d+1$. 
Thus,~$E_X\cap D=\emptyset$. 
Furthermore, observe that for each edge~$e\in E_Y$ we have~$d(e)=2k+1$. 
Hence, for the set~$E_D:=D\cap E_Y$ we conclude that~$|E_D|\le k^2$. 
Next, we define the vertex sets~$S_X := \{x\mid \{s,x\}\in D\}$, the set of endpoints of edges in~$D$ incident with~$s$, and~$S_Y := \{y\mid \{y,t\}\in D\}$, the set of endpoints of edges in~$D$ incident with~$t$. 
In the following, we describe an~$(s,t)$-cut~$A$ for~$G'$ that avoids~$D$. 
We partition~$A$ into two sets~$A_0$ and~$A_1$, where~$A_0:=(E_Y\cup F)\setminus D$ and~$A_1:=\{\{x,w_{xy}\}\mid \{w_{xy},y\}\in E_D\}$. 
Next, we show that~$A$ is an~$(s,t)$-cut for~$G'$: 

Observe that every~$(s,t)$-path contains at least one subpath~$(x,w_{xy},y)$ for a vertex~$x\in X$ and a vertex~$y\in Y$ as an induced subgraph. If~$\{w_{xy},y\}\in E_D$ then~$\{x,w_{xy}\}\in A_1$, and otherwise if~$\{w_{xy},y\}\notin D_S$ then~$\{w_{xy},y\}\in A_0$. 
Furthermore, observe that the~$(s,t)$-cut~$A$ has capacity~$\omega(A_0) + \omega(A_1) = \omega(A_1) = |E_D|\le k^2$. 
Since~$D$ is a solution of~$I'$, we conclude that~$A$ has capacity at least~$a+1 = k^2$ and, thus,~$|E_D| = \omega(A) = k^2$.
Thus,~$|D\cap F|\le 2k$. 

Next, assume towards a contradiction that the set~$E_D$ contains an edge~$\{w_{xy},y\}$ such that~$x\notin S_X$ or~$y\notin S_Y$. 
Without loss of generality assume that~$y\notin S_Y$. 
We set~$A^* := A \setminus \{\{x, w_{xy}\}\}$ and show that~$A^*$ is an~$(s,t)$-cut in~$G'$.
Since~$\{y,t\}\in A^*$ and for each~$x' \in N_G(y)\setminus \{x\}$ either~$\{x', w_{x'y}\}\in A^*$ or~$\{w_{x'y}, y\}\in A^*$, we obtain that~$A^*$ is an~$(s,t)$-cut of capacity~$\omega(A) - \omega(\{x, w_{xy}\}) = k^2-1 = a$.
A contradiction. 
Hence, for each edge~$\{w_{xy},y\}\in D$ we have~$\{s,x\}\in D$ and~$\{y,t\}\in D$. 
Since~$|E_D|=k^2$, we conclude that~$|S_X|=k=|S_Y|$.
Consequently,~$(S_X,S_Y)$ is a~$(k,k)$-biclique in~$G$ and, thus,~$I$ is a yes-instance of~\BIC.
\lncsqed\end{proof}

Together with~\Cref{cor-weighted-to-unweigthed} we obtain the following.

\begin{corollary}
\ZMCI{} is \W1-hard when parameterized by~$d+a$ even if~$c(e) = 1$ and~$\omega(e) \in \{0,1\}$ for all~$e\in E$.
\end{corollary}
\fi

\section{Parameterization by Vertex Cover Number}

We investigate the parameterization by the vertex cover number~$\vc(G)$.
Observing that for \MCI the number of protected edges~$d$ is at most~$2\vc(G)$ in nontrivial instances, eventually leads to the following FPT result.

\begin{theorem}\label{thm-mci-fpt-vc}
\MCI can be solved in~$2^{\Oh(\vc(G)^2)} \cdot n^{\Oh(1)}$~time.
\end{theorem}
\iflong
\begin{proof}
Let~$J:=(G,s,t,d,a)$ be an instance of~\MCI, and let~$\vc$ be the size of a minimum vertex cover of~$G$. The algorithm that we describe here is based on two observations which we formalize in two claims. The first claim states that the defender budget~$d$ is upper bounded by~$2 \cdot \vc$.

\begin{myclaim} \label{Claim: vc bound for d}
If~$d \geq 2 \cdot \vc$, then~$J$ is a yes-instance.
\end{myclaim}
\iflong
\begin{claimproof}
Let~$S$ be a minimum vertex cover in~$G$ and let~$P$ be a shortest~$(s,t)$-path in~$G$. Then, for every pair~$v,w$ of consecutive vertices on~$P$ at least one of~$v$ and~$w$ is contained in~$S$. Consequently, there are at most~$2\cdot \vc(G)$ edges on~$P$. If~$d \geq 2 \cdot \vc(G)$, then the set of edges on~$P$ is a solution of~$J$ of size at most~$d$. Thus,~$J$ is a yes-instance.
\end{claimproof}
\fi

Let~$S$ be a minimum vertex cover in~$G$, and let~$I:=V \setminus S$ be the remaining independent set. With the next claim we state that only a bounded number of vertices in~$I$ is needed to find a minimal solution of~$J$. To this end we introduce some notation: Given a subset~$X \subseteq S$, we let~$I_X:=\{u \in I \mid N_G(u)=X\} \subseteq I$ denote the \emph{neighborhood class} of~$X$. Moreover, we let~$E_X$ denote the set of edges between~$X$ and~$I_X$.

\begin{myclaim} \label{Claim: Small Neighborhood Classes}
There exists a minimum solution~$D$ of~$J$ such that~$|D \cap E_X| \leq |X|$ for every~$X \subseteq D$.
\end{myclaim}
\iflong
\begin{claimproof}
Let~$D$ be a solution of~$J$. If~$|D \cap E_X| \leq |X|$ for every~$X \subseteq S$, nothing more needs to be shown. 
Thus, consider some~$X \subseteq S$ such that~$|D \cap E_X| > |X|$, and let~$u \in I_X$. 
We then define~$D':=(D \setminus E_X) \cup E(u,X)$. It then holds that~$|D' \cap E_X| = |X|$. 
Moreover, observe that~$|D'|<|D|$ and~$D'\setminus E_X = D \setminus E_X$.

We next show that~$D'$ is a solution of~$J$. 
Let~$A \subseteq E \setminus D'$ be an~$(s,t)$-cut in~$G$ that is minimum among all~$(s,t)$-cuts that avoid~$D'$. We first prove that~$A \cap E_X = \emptyset$. Obviously,~$E(u,X) \cap A = \emptyset$ since~$E(u,X) \subseteq D'$. Consider~$u' \in I_X \setminus \{u\}$. Then, since~$N(u')=X$, for every~$(s,t)$-path~$P'$ containing~$u$, there are two consecutive edges~$\{x_1,u'\}$ and~$\{u',x_2\}$ with~$x_1,x_2 \in X$ on~$P'$. Since~$N(u')=N(u)$, replacing~$u'$ with~$u$ defines another~$(s,t)$-path~$P$ in~$G'$. Then, since~$\{x_1,u\}$ and~$\{u,x_2\}$ are not contained in~$A$, there exists another edge on~$P$ that is an element of~$A$. Consequently, on every~$(s,t)$-path~$P'$ containing~$u'$, there exists an edge in~$A$ that is not an element of~$E(u',X)$. Then, the fact that~$A$ is a minimum~$(s,t)$-cut among all~$(s,t)$-cuts that avoid~$D'$ implies~$E(u',X) \cap A = \emptyset$. Therefore,~$A \cap E_X= \emptyset$.


Then, since~$A \cap E_X = \emptyset$ and~$D' \setminus E_X = D \setminus E_X$ we have~$A \subseteq E \setminus D$. 
Consequently,~$|A|>a$ since~$D$ is a solution of~$J$.

Since~$D' \setminus E_X = D \setminus E_X$, the modification of~$D$ described above can be applied on all neighborhood classes~$I_X$ independently. Therefore, there exists a minimum solution of~$J$ that has the described property.
\end{claimproof}
\fi

Let~$X \subseteq S$ and~$I_X:=\{v_1, \dots, v_{|I_X|}\}$. We define~$I_X'$ by~$I_X':= I_X$ if~$|I_X| \leq |X|$ and~$I_X':=\{v_1, \dots, v_{|X|}\}$, otherwise. Due to Claim~\ref{Claim: Small Neighborhood Classes}, there exists a minimum solution such that at most~$|X|$ vertices in~$I_X$ are endpoints of edges in~$S$. Without loss of generality we may assume that all of these endpoints are from~$I_X'$. Thus, we can assume that there is a minimum solution~$D \subseteq E(S) \cup \bigcup_{X \subseteq S} E(X,I'_X)$. We use this assumption for the algorithm that we describe as follows.

\begin{enumerate}
\item If~$d \geq 2 \cdot \vc(G)$, then return~\emph{yes}.
\item Otherwise, we compute a minimum vertex cover~$S$. Iterate over every possible edge-set~$D \subseteq E(S \cup \bigcup_{X \subseteq S} E(X,I'_X)$ with~$|D| \leq d$, and check with the algorithm behind Lemma~\ref{Lemma: Compute Cut} that every~$(s,t)$-cut~$A \subseteq E \setminus S$ in~$G$ has size bigger than~$a$. If this is the case, then return~\emph{yes}.
\item If for none of the choices of~$D$ the answer~\emph{yes} was returned in Step 2, then return~\emph{no}.
\end{enumerate}

The correctness of the algorithm is implied by Claims~\ref{Claim: vc bound for d} and~\ref{Claim: Small Neighborhood Classes}. It remains to analyze the running time. Obviously, Step~1 and Step~3 can be performed in linear time. Consider Step~2. A minimum vertex cover can be computed in~$\Oh(1.28^\vc + n \cdot \vc)$~time~\cite{CKX10}. Next, observe that
\begin{align*}
|E(S \cup \bigcup_{X \subseteq S} I'_X)| & \leq \vc^2 + \sum_{i=0}^{\vc} \binom{\vc}{i} i^2\\
& \leq \vc^2 (1+2^{\vc-1}).
\end{align*} 

Since~$d < 2 \cdot \vc$, there are less than~$(\vc^2 (1+2^{\vc-1}))^{2\vc}$ possible subsets~$D \subseteq E(S \cup \bigcup_{X \subseteq S} I'_X)$ with~$|D|\leq d$. Together with the running time from Lemma~\ref{Lemma: Compute Cut}, Step~2 can be performed in~$(\vc^2 (1+2^{\vc-1}))^{2\vc} \cdot n^{\Oh(1)}$ time. Altogether, the algorithm runs within the claimed running time.
\lncsqed\end{proof}
\fi

Theorem~\ref{thm:tw und a} implies that \WMCI can be solved in pseudopolynomial time on graphs with a constant treewidth and therefore on graphs with a constant vertex cover number. 
With the next two theorems we show that significant improvements of this result are presumably impossible.  
\begin{theorem}\label{theo:wmci paraNPhard vc}
\WMCI{} is weakly \NP-hard \iflong on graphs   with a vertex cover of size two \fi.
\end{theorem}
\iflong
\begin{proof}
We describe a polynomial time reduction from~\KNAP which is known to be weakly \NP-hard~\cite{GJ79}.

\prob{\KNAP}{A set~$U$, a size function~$f: U \to \mathds{N}$, a value function~$g: U \to \mathds{N}$, and two budgets~$B, C \in \mathds{N}$.}{Is there a set of items~$S \subseteq U$ such that~$f(S) := \sum_{u\in S} f(u) \leq B$ and~$g(S) := \sum_{u\in S} g(u) \geq C$?}

Let~$I:=(U, f,g,B,C)$ be an instance of~\KNAP.
We describe how to construct an equivalent instance~$I':=(G=(V,E), s,t,c,\omega, d,a)$ of~\WMCI where~$G$ has a vertex cover of size two in polynomial time.

We set~$V := U \cup \{s,t\}$ and~$E := \{\{s,u\},\{u,t\}\mid u \in U\}$.
Note that~$\{s,t\}$ is a vertex cover of size two in~$G$.
Next, we set~$d := B$ and~$a := |U| + C -1$.
Finally, for every~$u\in U$ we set~$c(\{s,u\}) := f(u)$, $\omega(\{s,u\}) := 1$, $c(\{u,t\}) := d+1$, and~$\omega(\{u,t\}) := g(u)+1$.

Next, we show that~$I$ is a yes-instance of~\KNAP if and only if~$I'$ is a yes-instance of~\WMCI.

$(\Rightarrow)$
Suppose that~$I$ is a yes-instance of~\KNAP.
Then, there is a set~$S_U \subseteq U$ such that~$f(S_U) \leq B = d$ and~$g(S_U) \geq C$.
We set~$D:=\{\{s, u\}\mid u \in S_U\}$.
By construction, we obtain that~$c(D) = f(S_U) \leq d$.
Let~$A \subseteq E\setminus D$ be an~$(s,t)$-cut. We show that $A$ has capacity larger than~$a$.

Since~$\{s,u\}\in D$ for all~$u\in S_U$ it holds that~$T := \{\{u,t\}\mid u\in S_U\}\subseteq A$.
Note that~$\omega(T) = \sum_{u\in S_U} (g(u) + 1) = g(S_U) + |S_U|$.
Moreover, because of the path~$(s,u,t)$ for every~$u\in U \setminus S_U$ we obtain that~$\{s,u\}\in A$ or~$\{u,t\}\in A$.
Since both of these edges have capacity at least one, we obtain~$\omega(A) \geq \omega(T) + |U\setminus S_U| = g(S_U) + |U| = a+1$.
Consequently,~$I'$ is a yes-instance of~\WMCI.

$(\Leftarrow)$
Suppose that~$I'$ is a yes-instance of~\WMCI.
Then, there is a solution~$D\subseteq E$ with~$c(D) \leq d$.
By the fact that~$c(\{u,t\}) = d+1$ for all~$u\in U$, it follows that~$D \subseteq \{\{s,u\}\mid u\in U\}$.

Let~$S_U := \{u\in U\mid \{s,u\}\in D\}$. 
By construction, $f(S_U) = c(D) \leq d = B$.
We show that~$g(S_U) \geq C$.
Let~$A \subseteq E\setminus D$ be an~$(s,t)$-cut of minimum capacity.
Recall that~$\omega(A) \geq a+1 = |U| + C$.
Since~$A$ is an~$(s,t)$-cut in~$G$ and disjoint to~$D$, we know that~$\{u,t\}\in A$ for all~$u\in S_U$. 
Moreover, since~$\omega(\{s,u\}) = 1 \leq \omega(\{u,t\})$ for all~$u\in U$, we can assume without loss of generality, that~$\{s,u\}\in A$ for all~$u\in U \setminus S_U$.
Hence,~$a+1 = |U| + C \leq \omega(A) = |U\setminus S_U| + \sum_{u\in S_U}(g(u) +1) = g(S_U) + |U|$. Thus,~$C \leq g(S_U)$.
Consequently,~$I$ is a yes-instance of~\KNAP.
\lncsqed\end{proof}
\fi

\begin{theorem}\label{thm:w hard vc poly weight}
\WMCI{} is \W1-hard when parameterized  by \iflong the vertex cover number\fi~$\vc(G)$ even if~$c(e) + \omega(e) \in n^{\Oh(1)}$ and the graph is a biclique.
\end{theorem}
\begin{proof}
We describe a parameterized reduction from~\BIN which is \W1-hard when parameterized by~$k$ even if the size of each item is polynomial in the input size~\cite{JKMS13}.

\prob{\BIN}{A set~$U$ of items, a size-function~$f : U \to \mathds{N}$, and integers~$B$ and~$k$.}{Is there a~$k$-partition~$(U_1, \dots, U_k)$ of~$U$ with~$\sum_{u\in U_i} f(u) = B$ for all~$i\in [1,k]$?}

Let~$I:=(U,f,B,k)$ be an instance of~\BIN where the size of each item is polynomial in the input size.
We can assume without loss of generality that~$\sum_{u\in U} f(u) = B k$, as, otherwise,~$I$ is a trivial no-instance of~\BIN.  
We construct an equivalent instance~$I':=(G=(V,E), s,t,c,\omega, d,a)$ of~\WMCI where~$G$ has a vertex cover of size~$k+1$.
The graph~$G$ is a biclique with bipartition~$(\{s\} \cup \mathcal{B}, \{t\} \cup U)$ where~$\mathcal{B} := \{b_1, \dots, b_k\}$.
We set~$d := |U|$, and
\begin{align*}
c(e) :=
\begin{cases}
1 & \text{if }e\in \{\{u,b\} \mid u\in U, b \in \mathcal{B}\}\text{, and}\\
d+1 & \text{otherwise.}
\end{cases}
\end{align*}
Let~$\lambda := 2B\cdot |U|$, we set
\begin{align*}
\omega(e):=
\begin{cases} \lambda \cdot f(u) &\text{if }e=\{s,u\}\text{ with } u \in U,\\
\lambda \cdot B &\text{if }e=\{t,b\}\text{ with } b \in \mathcal{B}\text{, and}\\
1 &\text{otherwise.}
\end{cases}
\end{align*}
Finally, we set~$a := |U|\cdot k + \lambda (B k - 1)$.
This completes the construction of~$I'$. Figure~\ref{Figure: Bin Packing Reduction} shows an example of the construction.
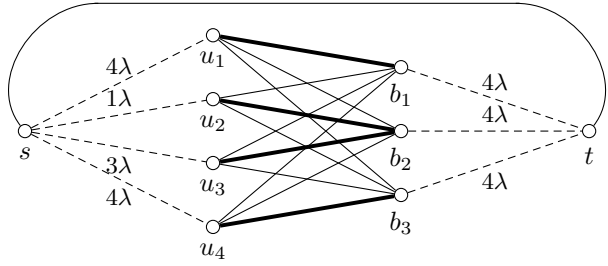
\begin{figure}[t]
\begin{center}
\begin{tikzpicture}[yscale=0.85]
\tikzstyle{knoten}=[circle,fill=white,draw=black,minimum size=5pt,inner sep=0pt]
\tikzstyle{bez}=[inner sep=0pt]

\node[knoten,label=below:$s$] (s)  at (-3,0) {};

\node[knoten,label=below:$u_1$] (u1)  at (-0.5,1.5) {};
\node[knoten,label=below:$u_2$] (u2)  at (-0.5,0.5) {};
\node[knoten,label=below:$u_3$] (u3)  at (-0.5,-0.5) {};
\node[knoten,label=below:$u_4$] (u4)  at (-0.5,-1.5) {};

\node[knoten,label=below:$b_1$] (B1)  at (2,1) {};
\node[knoten,label=below:$b_2$] (B2)  at (2,0) {};
\node[knoten,label=below:$b_3$] (B3)  at (2,-1) {};

\node[knoten,label=below:$t$] (t)  at (4.5,0) {};

\draw[-, densely dashed]  (s) to node[midway, above]{\small{$4 \lambda$}} (u1);
\draw[-, densely dashed]  (s) to node[above]{\small{$1 \lambda$}} (u2);
\draw[-, densely dashed]  (s) to node[below]{\small{$3 \lambda$}} (u3);
\draw[-, densely dashed]  (s) to node[midway,below]{\small{$4\lambda$}}(u4);

\draw[-, line width=1.5pt]  (u1) to (B1);
\draw[-]  (u1) to (B2);
\draw[-]  (u1) to (B3);
\draw[-]  (u2) to (B1);
\draw[-, line width=1.5pt]   (u2) to (B2);
\draw[-]  (u2) to (B3);
\draw[-]  (u3) to (B1);
\draw[-, line width=1.5pt]  (u3) to (B2);
\draw[-]  (u3) to (B3);
\draw[-]  (u4) to (B1);
\draw[-]  (u4) to (B2);
\draw[-, line width=1.5pt]  (u4) to (B3);

\draw[-, densely dashed]  (B1) to node[above]{\small{$4 \lambda$}} (t);
\draw[-, densely dashed]  (B2) to node[above]{\small{$4 \lambda$}} (t);
\draw[-, densely dashed]  (B3) to node[below]{\small{$4 \lambda$}} (t);

\node[bez] (s1)  at (-2,2) {};
\node[bez] (t1)  at (3.5,2) {};
\node[bez] (s2)  at (-2.1,2) {};
\node[bez] (t2)  at (3.6,2) {};

\draw[-, bend left=62]  (s) to (s1);
\draw[-]  (s2) to (t2);
\draw[-, bend right=62]  (t) to (t1);

\end{tikzpicture}
\caption{An example of the construction from the proof of Theorem~\ref{thm:w hard vc poly weight} for a \textsc{Bin Packing} instance with~$f(u_1)=f(u_4) = 4$,~$f(u_2)=1$,~$f(u_3)=3$, $B=4$, and~$k=3$. The thick edges represent a minimum solution~$D$. The edge-labels represent all edge capacities that are bigger than one. Observe that every~$(s,t)$-cut avoiding~$D$ contains dashed edges that have a capacity sum of at least~$12\lambda$. \label{Figure: Bin Packing Reduction}}
\end{center}
\end{figure}
Note that~$\{s\} \cup \mathcal{B}$ is a vertex cover of~$G$ of size~$k+1$.
It remains to show that~$I$ is a yes-instance of~\BIN if and only if~$I'$ is a yes-instance of~\WMCI.

$(\Rightarrow)$ 
Suppose that~$I$ is a yes-instance of~\BIN. 
Then, there is a~$k$-partition~$(U_1, \dots, U_k)$ of~$U$, such that~$\sum_{u\in U_i} f(u) = B$ for all~$i\in [1,k]$.
We set~$D:= \{\{u, b_i\}\mid i\in [1, k], u \in U_i\}$. Note that~$c(D) = d$. We next show that~$D$ is a solution.

Let~$A \subseteq E \setminus D$ be an~$(s,t)$-cut in~$G$ and let~$i\in [1,k]$.
Since for each~$u\in U_i$,~$D$ contains the edge~$\{u, b_i\}$, the~$(s,t)$-path~$P_u := (s, u, b_i, t)$ can only be cut if~$\{s,u\}\in A$ or~$\{b_i, t\}\in A$.
Consequently,~$\{\{s,u\}\mid u\in U_i\}  \subseteq A$ or~$\{b_i, t\}\in A$.
Recall that~$\sum_{u\in U_i}f(u) = B$.
Hence,~$\sum_{u\in U_i} \omega(\{s,u\}) = \sum_{u\in U_i}\lambda f(u) = \lambda B = \omega(\{b_i, t\})$.
Since~$P_u$ and~$P_w$ are edge-disjoint if~$u$ and~$w$ are in distinct parts of the~$k$-partition, we obtain that~$\omega(A) \geq k \lambda B> a$ and, thus,~$I'$ is a yes-instance of~\WMCI.

$(\Leftarrow)$
Suppose that~$I'$ is a yes-instance of~\WMCI.
Then, there is a solution~$D\subseteq E$ with~$c(D) \leq d$.
By construction,~$D \subseteq E(U, \mathcal{B})$, since all other edges have cost~$d+1$.

Note that for each~$u\in U$, there is some~$b\in \mathcal{B}$, such that~$\{u, b\}\in D$, as, otherwise~$A :=\{\{s,t\}\}\cup \{\{s, u'\} \mid u' \in U\setminus\{u\}\} \cup \{\{u, b'\}\mid b' \in \mathcal{B}\}\}$ is an~$(s,t)$-cut in~$G$ with capacity~$\lambda (B k - f(u)) + k +1< a$.
Since~$|D| \leq d$, we obtain that for each~$u\in U$, there is exactly one~$b\in \mathcal{B}$, such that~$\{u,b\}\in D$.

We set~$U_i := \{u \in U \mid \{u, b_i\}\in D\}$ for all~$i\in [1,k]$.
By the above, we obtain that~$(U_1, \dots, U_k)$ is a~$k$-partition of~$U$.
We show that~$\sum_{u\in U_i}f(u) = B$ for all~$i\in [1,k]$.

Assume towards a contradiction that~$\sum_{u\in U_i}f(u) \neq B$ for some~$i\in [1,k]$. 
This is the case if and only if there is some~$j\in [1,k]$ with~$\sum_{u\in U_j}f(u) < B$.
We set~$A := \{\{s,t\}\} \cup \{\{s, u\}\mid u\in U_j\} \cup \{\{b, t\}\mid b\in \mathcal{B} \setminus \{b_j\}\} \cup (E(U, \mathcal{B})\setminus D)$.
Note that~$\omega(A) = 1 + \lambda (\sum_{u\in U_j} f(u)) + \lambda B (k-1) + |U| \cdot (k-1)\le \lambda (B-1)  +\lambda B(k-1)+|U| \cdot k = \lambda (B k -1) + |U| \cdot k = a$, since~$\sum_{u\in U_j} f(u) < B$.
It remains to show that~$A$ is an~$(s,t)$-cut in~$G$.
Observe that~$N_{G-A}(t)=b_j$.
Since~$N_{G-A}(b_j)=\{t\}\cup U_j$ and~$N_{G-A}(u)=\{b_j\}$ for each~$u\in U_j$, we conclude that~$A$ is indeed an~$(s,t)$-cut in~$G$.
This contradicts the fact that there is no~$(s,t)$-cut disjoint to~$D$ in~$G$ of capacity at most~$a$.
As a consequence,~$\sum_{u\in U_i}f(u) = B$ for all~$i\in [1,k]$ and, thus,~$I$ is a yes-instance of~\BIN.
\lncsqed\end{proof}

We use \Cref{thm:w hard vc poly weight} to extend the \W1-hardness to~$\pw(G)+\fvs(G)$ and thus~$\vc(G)$ in the running time stated in \Cref{thm-mci-fpt-vc} can presumably not be replaced by~$\pw(G)+\fvs(G)$.
Hence, \MCI is also \W1-hard parameterized by~$\tw(G)$.

\begin{theorem}\label{theo: mci w hard pw}
\MCI is \W1-hard when parameterized by~$\pw(G) + \fvs(G)$\iflong , where~$\pw(G)$ denotes the pathwidth of the input graph and~$\fvs(G)$ denotes the size of the smallest feedback vertex set of the input graph\fi.
\end{theorem}
\iflong
\begin{proof}
We reduce from \WMCI{} which, due to~\Cref{thm:w hard vc poly weight}, is \W1-hard when parameterized by the vertex cover number~$\vc(G)$ even if~$c(e) + \omega(e) \in n^{\Oh(1)}$ and the graph is a biclique.

Let~$I=(G=(V,E),s,t,c,\omega,d,a)$ be an instance of~\WMCI where~$c(e) + \omega(e) \in n^{\Oh(1)}$ and the graph is a biclique.
Moreover, let~$I'=(G'=(V',E'),s,t,d,a)$ be the equivalent instance of~\MCI we obtain  obtain in polynomial time by applying the construction leading to~\Cref{cor-weighted-to-unweigthed}.
We show that both the size of the smallest feedback vertex set and the pathwidth of~$G'$ are upper-bounded by a function only depending on~$\vc(G)$.

Let~$(X,Y)$ be the bipartition of~$G$ and let~$X$ be the smaller part.
Thus,~$\vc(G) = |X|$.
Moreover, let~$x_1, \dots, x_{|X|}$ be the elements of~$X$ and let~$y_1, \dots, y_{|Y|}$ be the elements of~$Y$.
Recall that we obtain~$I'$ by replacing every edge~$e=\{u,v\} \in E$ by a subgraph~$G_e=(V_e, E_e)$ which consists of vertex disjoint~$(u,v)$-paths (besides~$u$ and~$v$).

Note that~$G_e$ has a path decomposition~$\mathcal{B}_e$ of width at most three where each bag contains both endpoints of~$e$. 
For every~$y\in Y$ we set~$\mathcal{B}_y := \mathcal{B}_{\{x_1, y\}} \cdot \ldots \cdot  \mathcal{B}_{\{x_{|X|}, y\}}$.
Note that~$\mathcal{B}_y$ is a path decomposition of width at most three for~$G_y = (\bigcup_{x\in X} V_{\{x,y\}}, \bigcup_{x\in X} E_{\{x,y\}})$.

Finally, let~$\mathcal{B} := \mathcal{B}_{y_1} \cdot \dots \cdot \mathcal{B}_{y_{|Y|}}$ and let~$\mathcal{B}'$ be the sequence of bags we obtain from~$\mathcal{B}$ by adding all vertices of~$X$ to each of the bags of~$\mathcal{B}$.
By construction,~$\mathcal{B}'$ is a path decomposition of width at most~$|X| + 3$ for~$G'$.
Hence,~$\pw(G') \leq |X| + 3 = \vc(G) + 3$.

It remains to show that~$\fvs(G') \leq \vc(G)$.
Note that~$G_{\{x,y\}} - \{x\}$ is acyclic for each~$x\in X$ and~$y\in Y$.
Hence,~$G' - X$ is acyclic since~$Y$ is an independent set in~$G$ and for each pair of distinct edges~$e_1, e_2\in E$ it holds that~$V_{e_1} \cap V_{e_2}  = e_1 \cap e_2$.
Consequently,~$\fvs(G') \leq \vc(G)$ and, thus, \MCI is \W1-hard when parameterized by~$\pw(G') + \fvs(G')$.
\lncsqed\end{proof}
\fi

\section{On Problem Kernelization}

\iflong
\subsection{A Polynomial Kernel for~$\vc+a$}
\fi

On the positive side, we show that~\WMCI admits a polynomial kernel when parameterized by~$\vc+a$. The main tool for this kernelization is the merge of vertices according to Definition~\ref{def:edge contract}. 

Let~$J:=(G=(V,E),s,t,c,\omega,d,a)$ be an instance of~\WMCI. 
We first provide two simple reduction rules that remove degree-one vertices.

\begin{redrule} \label{Rule: st deg one}
If~$s$ has exactly one neighbor~$w$ and~$\omega(\{s,w\}) \leq a$, then delete~$s$, set~$s:=w$, and decrease~$d$ by~$c(\{s,w\})$. Analogously, if~$t$ has exactly on neighbor~$v$ and~$\omega(\{t,v\}) \leq a$, then delete~$t$, set~$t:=v$, and decrease~$d$ by~$c(\{t,v\})$.
\end{redrule}

\iflong
The safeness of Rule~\ref{Rule: st deg one} follows by the observation that, if~$s$ (or~$t$, respectively) is incident with a unique edge~$e$ with~$\omega(e) \leq a$, this edge must be part of every solution, since~$M:=\{e\}$ is an~$(s,t)$-cut of capacity at most~$a$.
\fi
\begin{redrule} \label{Rule: Delete deg one}
If there exists a degree-one vertex~$v \not \in \{s,t\}$, then delete~$v$.
\end{redrule}

\iflong
It is easy to see that Rule~\ref{Rule: Delete deg one} is safe. Since~$v \neq s$ and~$v \neq t$, the single edge incident with~$v$ is not contained in any inclusion-minimal~$(s,t)$-cut and therefore not part of any minimal solution. 
\fi
The next reduction rule is the main idea behind the problem kernelization.

\begin{redrule} \label{Rule: merge vertices with big cut}
If there are vertices~$u,v \in V$ such that a minimum~$(u,v)$-cut has capacity at least~$a+1$, then merge~$u$ and~$v$.
\end{redrule}

\iflong
\begin{lemma}
Rule~\ref{Rule: merge vertices with big cut} is safe.
\end{lemma}

\begin{proof}
Recall that due to Lemma~\ref{lem: every edge in minimal cut} we can safely merge edges that are not contained in any inclusion-minimal~$(s,t)$-cut of capacity at most~$a$. We show the safeness of Rule~\ref{Rule: merge vertices with big cut} by applying Lemma~\ref{lem: every edge in minimal cut} on the vertex pair~$\{u,v\}$. Note that~$u$ and~$v$ are not necessarily adjacent in~$G$. Thus, we first transform~$J$ into an instance~$J'$ by adding an edge~$\{u,v\}$ with cost~$d+1$ and capacity~$a+1$. Let~$c'$ and~$\omega'$ be the cost functions and capacity functions of~$J'$. \todom[inline]{satz komisch}
We next show that~$J$ is a yes-instance if and only if~$J'$ is a yes-instance.

$(\Rightarrow)$ Let~$S$ be a solution of~$J$ with~$c(S) \leq d$. Since adding an edge might only increase the size of a cut,~$S$ is a solution of~$J'$. 

$(\Leftarrow)$ Let~$S'$ be a solution of~$J'$ with~$c'(S) \leq d$. Then, $\{u,v\} \not \in S'$ since~$c'(\{u,v\})=d+1$. We show that~$S'$ is a solution of~$J$. Let~$M \subseteq E \setminus S'$ be an inclusion-minimal~$(s,t)$-cut in~$G$. 
We consider the corresponding partition~$(A,B)$ of~$V$. If~$u \in A$ and~$v \in B$ or vice versa, then~$M$ is an~$(u,v)$-cut in~$G$ and therefore~$\omega(M) \geq a+1$ by the condition of Rule~\ref{Rule: merge vertices with big cut}. 
Otherwise, if~$u$ and~$v$ belong to the same partite set, then~$M$ is an~$(s,t)$-cut in~$G'$. Since~$S$ is a solution of~$J'$ we conclude~$\omega(M) \geq a+1$.

Thus, the instances~$J$ and~$J'$ are equivalent. Note that~$\omega'(\{u,v\})=a+1$ implies that~$\{u,v\}$ is not contained in any inclusion-minimal~$(s,t)$-cut of capacity at most~$a$. Then, Lemma~\ref{lem: every edge in minimal cut} implies that~$u$ and~$v$ can safely be merged, which proves the safeness of Rule~\ref{Rule: merge vertices with big cut}.
\lncsqed\end{proof}
\fi
We now assume that~$J$ is reduced regarding Rules~\ref{Rule: st deg one}--\ref{Rule: merge vertices with big cut}.\todom{anpassen nach verschieben}
\iflong 
Before we show that the number of edges in~$G$ is at most~$2\cdot \vc(G) \cdot a$, we observe that there is no degree-one vertex in~$G$: Since~$J$ is reduced regarding Rule~\ref{Rule: Delete deg one}, every vertex in~$V\setminus \{s,t\}$ has degree at least two. Furthermore, since~$J$ is reduced regarding Rule~\ref{Rule: st deg one} the vertices~$s$ and~$t$ are not incident with a unique edge of capacity at most~$a$. \todom[inline]{auch seltsamer satz}
Finally, since~$J$ is reduced regarding Rule~\ref{Rule: merge vertices with big cut}, the vertices~$s$ and~$t$ are not incident with a unique edge~$\{s,u\}$ (or~$\{t,v\}$, respectively) of capacity at least~$a+1$ since the vertices~$u$ and~$v$ (or~$t$ and~$v$) would have been merged by Rule~\ref{Rule: merge vertices with big cut}.

To show that the number of edges in~$J$ is at most~$2 \cdot \vc(G) \cdot a$, we introduce \emph{cut trees} which are special binary trees. Throughout this section, given an inner vertex~$x$ of a binary tree, we let~$x_\ell$ denote its left child and~$x_r$ denote its right child.

\begin{definition} \label{Definition: Cut-Tree}
Let~$G=(V,E)$ be a graph with a capacity function~$\omega:E \rightarrow \mathds{N}$, let $S \subseteq V$ be a vertex cover of~$G$. Let~$T=(\mathcal{V},\mathcal{E})$ be a binary tree with root vertex~$r \in \mathcal{V}$ and~$\psi:\mathcal{V} \rightarrow 2^V$. Then,~$(T,\psi)$ is a~\emph{cut tree of~$G$ with respect to~$S$}~if
\begin{enumerate}
\item $\psi (r)=V$,
\item for every vertex~$x \in \mathcal{V}$ with~$|\psi(x) \cap S| \geq 2$, there exist vertices~$u,v \in \psi(x) \cap S$ and a minimum~$(u,v)$-cut~$M$ in~$G[\psi(x)]$ with partitions~$(A,B)$ such that~$\psi(x_\ell)=A$ and~$\psi(x_r)=B$, and
\item every vertex~$x \in \mathcal{V}$ with~$|\psi(x) \cap S| = 1$ is a leaf.
\end{enumerate}
\end{definition}


Recall that we consider a reduced instance~$J$ with input graph~$G$. In the following, let~$S$ be a minimum vertex cover of~$G$. We consider a cut tree~$(T,\psi)$ of~$G$ with respect to~$S$. Observe that there is no inner vertex of~$T$ that has exactly one child, and that~$\{\psi(x) \mid x\text{ is a leaf of }T\}$ is a partition of~$V$, where each set of the partition contains exactly one vertex from~$S$. Thus, if~$S$ is a minimum vertex cover, then T consists of at most~$\vc(G)$ inner vertices and~$\vc(G)$ leaves. Furthermore, note that for each inner vertex~$x$, the tuple~$(\psi(x_\ell),\psi(x_r))$ is a partition of~$\psi(x)$. 

To give a bound on the number of edges of~$G$, we associate an edge-set~$E_x$ with every~$x \in \mathcal{V}$. If~$x$ is an inner vertex in~$T$, then we define~$E_x:=E_G(\psi(x_\ell),\psi(x_r))$. Otherwise, if~$x$ is a leaf, then we define~$E_x := E_G(\psi(x))$. Observe that for every inner vertex~$x$ the edge-set~$E_x$ is a minimum~$(u,v)$-cut in~$G$ for a pair of vertices~$u,v \in \psi(x)$. The size bound of the number of edges mainly relies on the following lemma.

\begin{lemma} \label{Lemma: cut tree union is E}
 Let~$(T=(\mathcal{V},\mathcal{E}),\psi)$ be a cut tree of~$G=(V,E)$. Then,
$E= \bigcup_{x \in \mathcal{V}} E_x$.
\end{lemma}

\begin{proof}
It clearly holds that~$\bigcup_{x \in \mathcal{V}} E_x \subseteq E$ since each~$E_x \subseteq E$. It remains to prove~$E \subseteq \bigcup_{x \in \mathcal{V}} E_x$. 

Let~$e=\{u,v\} \in E$. If~$e \in E_x$ for some leaf vertex~$x$, nothing more needs to be shown. Otherwise, consider the leaf vertices~$x$ and~$y$ with~$u \in \psi(x)$ and~$v \in \psi(y)$ and let~$z$ be the first common ancestor of~$x$ and~$y$. Then,~$u \in \psi(z_\ell)$ and~$v \in \psi(z_r)$ or vice versa. Consequently, $e \in E_G(\psi(z_\ell),\psi(z_r)) = E_z$.
\lncsqed\end{proof}

We now prove the main result of this subsection.
\fi

\begin{theorem} \label{Theorem: Kernel a+vc}\iflong
There is an algorithm that, given an instance of~\WMCI computes an equivalent instance in polynomial time, such that the graph consists of at most~$2  \vc(G) \cdot a$~edges.
\fi
\end{theorem}
\iflong
\begin{proof}
The algorithm is simply described as follows: Apply the Rules~\ref{Rule: st deg one}--\ref{Rule: merge vertices with big cut} exhaustively. Obviously, a single application of one rule can be done in polynomial time. Then, since after every application of one of the rules the number of vertices is decreased by one, Rules~\ref{Rule: st deg one}--\ref{Rule: merge vertices with big cut} can be applied exhaustively in polynomial~time.

Let~$J$ be an instance of~\WMCI that is reduced regarding Rules~\ref{Rule: st deg one}--\ref{Rule: merge vertices with big cut}. We next use Lemma~\ref{Lemma: cut tree union is E} to prove that the input graph~$G$ consists of at most~$2 \cdot \vc \cdot a$ edges. Recall that for every pair~$(u,v)$ of vertices in~$G$, there exists a~$(u,v)$-cut of size at most~$a$ since~$J$ is reduced regarding Rule~\ref{Rule: merge vertices with big cut}

Let~$(T=(\mathcal{V},\mathcal{E}),\psi)$ be a cut tree of~$G$ with respect to a minimum vertex cover~$S$. Let~$I := V \setminus S$ be the remaining independent set. Furthermore, let~$\mathcal{L} \subseteq \mathcal{V}$ be the set of leaves of~$T$ and let~$\mathcal{I} \subseteq \mathcal{V}$ be the set of inner vertices of~$T$. Lemma~\ref{Lemma: cut tree union is E} then implies
\begin{align*}
|E| \leq |\bigcup_{x \in \mathcal{I}} E_x| + |\bigcup_{x \in \mathcal{L}} E_x|.
\end{align*}
Since every~$(u,v)$-cut in~$G$ has size at most~$a$ and~$\omega(e) \geq 1$ for every edge~$e$ we conclude that~$|E_x| \leq a$ for every~$x \in \mathcal{I}$. Thus, since~$T$ has at most~$\vc(G)$ inner vertices, we have~$|\bigcup_{x \in \mathcal{I}} E_x| \leq \vc(G) \cdot a$.

We next define an injective mapping~$p: \bigcup_{x \in \mathcal{L}} E_x \rightarrow \bigcup_{x \in \mathcal{I}} E_x$. Observe that the existence of such a mapping implies~$|\bigcup_{x \in \mathcal{L}} E_x| \leq |\bigcup_{x \in \mathcal{I}} E_x|$ and thus~$|E| \leq 2 \cdot \vc(G) \cdot a$.

Let~$\{u,v\} \in E_x$ for some leaf vertex~$x$. Without loss of generality assume that~$v \in S$ and~$u \in I$. Since~$J$ is reduced regarding Rules~\ref{Rule: st deg one}--\ref{Rule: merge vertices with big cut}, there are no degree-one vertices in~$G$ and thus,~$u$ has a neighbor~$w \in S\setminus \{v\}$. We then define~$p(\{u,v\}):=\{u,w\}$. Note that~$p(\{u,v\}) \in E_G(S,I)$, and that both edges $\{u,v\}$ and~$p(\{u,v\})$ are incident with the same vertex~$u \in I$.

We first show that~$p$ is well-defined. That is, that~$p(\{u,v\}) \in \bigcup_{x \in \mathcal{I}} E_x$ for every~$\{u,v\} \in \bigcup_{x \in \mathcal{L}} E_x$. Since~$|\psi(x) \cap S|=1$, there exists another leaf vertex~$y$ with~$w \in \psi(y)$. Let~$z$ be the first common ancestor of~$x$ and~$y$. Then,~$\{u,w\} \in E_z$. Since~$z$ is an inner vertex, we conclude that~$p$ is well-defined.

Next, we show that~$p$ is injective. Let~$e:=\{u,v\}$ and~$e':=\{u',v'\}$ be edges in~$\bigcup_{x \in \mathcal{L}} E_x$. Let~$p(e)=p(e')$. We show that~$e=e'$. Without loss of generality assume that~$v,v' \in S$ and~$u,u' \in I$. Then, all four edges~$e$, $e'$, $p(e)$, and~$p(e')$ are incident with the same vertex of~$I$ and thus~$u=u'$. Then, since~$\{\psi(x) \mid x \in \mathcal{L}\}$ is a partition of~$V$ we conclude that~$e$ and~$e'$ are element of the same set~$E_x$ for some~$x \in \mathcal{L}$. Then,~$|\psi(x) \cap S|=1$ implies~$v=v'$ and thus~$e=e'$. Therefore,~$p$ is injective, which then implies~$|E| \leq 2 \cdot \vc(G) \cdot a$.
\lncsqed\end{proof}

Technically, the instance from Theorem~\ref{Theorem: Kernel a+vc} is not a kernel since the encoding of~$d$ and the values of~$c(e)$ might not be bounded by some polynomial in~$a$ and~$\vc$. We use the following lemma to show that Theorem~\ref{Theorem: Kernel a+vc} implies a polynomial kernel for~\WMCI.

\begin{lemma}[\cite{EKMR17}] \label{Lemma: shrink score}
There is an algorithm that, given a vector~$w \in \mathds{Q}^r$ and some~$W \in \mathds{Q}$ computes in polynomial time a vector~$\overline{w}=(w_1, \dots, w_r) \in \mathds{Z}^r$ where~$\max_{i\in \{1, \dots, r\}} |w_i| \in 2^{\Oh(r^3)}$ and an integer~$\overline{W} \in \mathds{Z}$ with total encoding length~$\Oh(r^4)$ such that~$w \cdot x \leq W$ if and only if~$\overline{w} \cdot x \leq \overline{W}$ for every~$x \in \{0,1\}^r$.
\end{lemma}

\begin{corollary}
\WMCI admits a polynomial problem kernel when parameterized by~$\vc + a$.
\end{corollary}

\begin{proof}
Let~$J:=(G=(V,E),s,t,c,\omega,d,a)$ be the reduced instance from Theorem~\ref{Theorem: Kernel a+vc}. Observe that both, the number of vertices~$n$ and the number of edges~$m$ of~$G$ are polynomially bounded in~$\vc+a$. We define~$r:=m$ and~$w$ to be the~$r$-dimensional vector where the entries are the values~$c(e)$ for each~$e \in E$. Furthermore, let~$W:=d$. Applying the algorithm behind Lemma~\ref{Lemma: shrink score} computes a vector~$\overline{w}$ with the property stated in the lemma and an integer~$\overline{W}$ that has encoding length~$\Oh(m^4)$.

Substituting all values~$c(e)$ with the corresponding entry in~$\overline{w}$ and substituting~$d$ by~$\overline{W}$ then converts~$J$ into an equivalent instance which has a size that is polynomially bounded in~$\vc+a$.
\lncsqed\end{proof}
\fi

\iflong
The algorithm behind Theorem~\ref{Theorem: Kernel a+vc} also implies a polynomial kernel for the unweighted problem~\MCI: We transform the unweighted instance into a weighted instance where all capacities and costs are one. Afterwards, we apply the algorithm from Theorem~\ref{Theorem: Kernel a+vc} to compute a reduced instance~$J'$. In~$J'$ all costs are one, and the capacities are at most~$a+1$. We then use Corollary~\ref{cor-weighted-to-unweigthed} to transform the reduced instance~$J'$ into an instance~$J$ of~\MCI. 
Due to the structure of~$J'$, the number of new vertices introduced in~$J$ is at most~$m \cdot (a+1)$, where~$m$ denotes the number of edges in~$J'$. Since~$m \le 2\vc(G)\cdot a$, we obtain the following corollary.
\fi
\begin{corollary}
\MCI admits a polynomial \iflong problem\fi kernel with~$4\vc(G) \cdot a^2$~edges.
\end{corollary}
\todom{prove the constant of the kernel}

\iflong
\subsection{Limits of Problem Kernelization}
\fi

\iflong
Let~$B_q$ be a full binary tree of height~$q$. 
We denote the vertices on level~$\ell$ as~$b_{\ell,1}^q, \ldots , b_{\ell,2^\ell}^q$ for each~$\ell\in[0,q]$.
Hence, vertex~$b_{\ell,i}^q$ for some~$\ell\in[0,q-1]$ and some~$i\in[1,2^\ell]$ has the neighbors~$b_{\ell+1,2i-1}^q$ and~$b_{\ell+1,2i}^q$ in the next level.
The full binary tree~$R_q$ of height~$q$ with the vertices~$r_{\ell,1}^q, \ldots , r_{\ell,2^\ell}^q$ on level~$\ell\in[0,q]$ is defined analogously. 
A \emph{mirror fully binary tree}~$M_q$ is the graph obtained after merging the vertices~$b_{q,i}^q$ and~$r_{q,i}^q$ for each~$i\in[1,2^q]$.
By~$\lp(G)$ we denote the length of a longest path in~$G$.  

\begin{lemma}
\label{lem-longest-path-in-mirror-fully-binary trees}
Let~$q\ge 3$, then the longest path of a mirror fully binary tree~$M_q$ is~$2q^2$.
\end{lemma}
\begin{proof}
By~$L_q$ we denote the length of each longest path with one endpoint being~$b^q_{0,1}$ which does not contain vertex~$r^q_{0,1}$.
We prove the following statements inductively for~$q\ge 3$.
\begin{enumerate}
\item $L_q=L_{q-1}+2q-1$.
\item $|V(P_q)\cap \{b^q_{0,1},r^q_{0,1}\}|=1$ for each longest path~$P_q$ of~$M_q$ and $\lp(M_q)=2L_q$.
\end{enumerate}

Solving the recurrence implied by~$1.$ and~$2.$ leads to~$\lp(M_q)=2q^2$.
Hence, it remains to prove the two statements.

\textbf{Base Case~$q=3$:}
By considering all possible longest paths in~$M_3$ we show that the length of a longest path with one endpoint being~$b^3_{0,1}$ and not containing~$r^3_{0,1}$ is nine and that~$\lp(M_3)=18$. 

\textbf{Inductive step~$j-1\mapsto j$:}

\begin{enumerate}
\item Let~$Z_j$ be a longest path starting at~$b^j_{0,1}$ not containing vertex~$r^j_{0,1}$.
Without loss of generality, assume that~$b^j_{1,1}\in Z_j$.

First, consider the case that~$r^j_{1,1}\notin Z_j$.
Then, the length of~$Z_j$ is at most the length of a longest path starting in vertex~$b^j_{1,1}$ and not containing vertex~$r^j_{1,1}$ in the mirror fully binary tree of height~$j-1$ rooted in vertex~$b^j_{1,1}$ plus one for the edge~$\{b^j_{1,1},b^j_{0,1}\}$.
By inductive hypothesis we obtain that~$Z_j$ has length at most~$L_{j-1}+1$.

Second, consider the case that~$r^j_{1,1}\in Z_j$.
Since~$b^j_{1,1},r^j_{1,1}\in Z_j$ there exists a path connecting these two vertices. 
Without loss of generality, assume that vertices~$b^j_{2,1}$ and~$r^j_{2,1}$ are on this path of length exactly~$2j-2$.
Note that~$b^j_{1,1}$ has the neighbors~$b^j_{0,1}$ and~$b^j_{2,1}$ in the path and thus~$b^j_{2,2}$ is not a neighbor of~$b^j_{1,1}$.
Thus, we can now use the inductive hypothesis.
The length of each longest path starting at vertex~$r^j_{1,1}$ and not containing vertex~$b^j_{1,1}$ in the mirror fully binary tree of height~$j-1$ rooted in~$r^j_{1,1}$ is at most~$L_{j-1}$.
Hence, the length of~$Z_j$ is at most~$L_{j-1}+2j-1$.
Thus, we obtain~$L_j=L_{j-1}+2j-1$.

\item 
Consider the case that~$|V(P_j)\cap \{b^j_{0,1},r^j_{0,1}\}|=1$.
Then, by~$1.$ we can construct a path~$Z_j$ of length~$2L_j=2L_{j-1}+4j-2$ by joining two paths where one endpoint is~$b^j_{0,1}$ which both do not contain vertex~$r^j_{0,1}$.
Thus,~$\lp(M_j)\ge 2L_j$.
Next, assume towards a contradiction that~$|V(P_j)\cap \{b^j_{0,1},r^j_{0,1}\}|\ne 1$.

First, consider that case~$|V(P_j)\cap \{b^j_{0,1},r^j_{0,1}\}|=0$. 
Then, the path~$P_j$ has length at most~$\lp(M_{j-1})=2L_{j-1}<2L_{j-1}+4j-2=2L_j$, a contradiction to the existence of the path~$Z_j$.

Second, consider the case~$|V(P_j)\cap \{b^j_{0,1},r^j_{0,1}\}|=2$.
Let~$Q_j$ be the unique subpath of~$P_j$ with endpoints~$b^j_{0,1}$ and~$r^j_{0,1}$. 
Without loss of generality,~$b^j_{1,1}\in V(Q_j)$ (and thus also~$r^j_{1,1}\in V(Q_j)$). 
Since~$M_j$ is a mirror fully binary tree the subpath~$Q_j$ has length exactly~$2j$.
Let~$M'$ be the mirror fully binary tree rooted at vertex~$b^j_{1,1}$.
Then~$M'\cap V(P_j)=V(Q_j)\setminus\{b^j_{0,1},r^j_{0,1}\}$.
Furthermore,~$P_j$ can contain the edges~$\{b^j_{0,1},b^j_{1,2}\}$ and~$\{r^j_{0,1},r^j_{1,2}\}$, and a longest path in the mirror fully binary tree rooted in~$b^j_{1,2}$ of height~$j-1$.
Thus, the length of~$P_j$ is at most~$\lp(M_{j-1})+2j+2=2L_{j-1}+2j+2<2L_{j-1}+4j-2=2L_j$, a contradiction to the existence of the path~$Z_j$.

Hence,~$\lp(M_j)=2L_j$.
\end{enumerate}\lncsqed\end{proof}
\fi 

On the negative side, we provide an OR-composition to exclude a polynomial kernel for the combination of almost all considered parameters \iflong with the exception of~$\vc(G)$ and~$\fvs(G)$\fi.
  
\begin{theorem}
\label{thm-mci-no-poly-kernel}
\iflong
None of the problems \MCI, \WMCI, and~\ZMCI admits a polynomial kernel when parameterized by~$d + a + \lp(G) + \Delta(G) + \td(G)$, unless~$\NP \subseteq \coNPpoly$, where~$\td(G)$ denotes the treedepth of~$G$.
\fi
\end{theorem}  
\iflong
\begin{proof}
Our strategy is as follows:
First, we provide an OR-composition~\cite{BDFH09,BJK14} of~$2^q$ instances of \MCI to \WMCI where~$\omega(e)=1$ and~$c(e)\in(d+q)^{\Oh(1)}$ for each edge~$e$.
Second, we apply Lemma~\ref{lem to unit costs} exhaustively to transform the constructed instance of \WMCI to an equivalent instance of \MCI.
Clearly, the budgets~$a$ and~$d$ do not change.
In this transformation, each edge~$e$ with~$c(e)\ge 2$ is replaced by a path with~$c(e)$ edges.
Hence, the maximum degree does not increase.
Furthermore, since~$c(e)\in(d+q)^{\Oh(1)}$, the length of the longest path does only increase by a factor of~$(d+q)^{\Oh(1)}$ and the tree-depth is only increased by~$\Oh(\log(d+q))$.
Thus, this transformation preserves all five parameters.
It remains to show the statement for \WMCI.

Now, we prove the no polynomial kernel result for \WMCI by presenting an OR-composition from~\MCI.
Let~$I_1, I_2, \dots, I_{2^q}$ be instances of~\MCI with the same budgets~$d$ and~$a$, the same maximum degree~$\Delta(G)$, the same length~$\lp(G)$ of the longest path, and the same tree-depth~$\td(G)$ for some integer~$q\ge 3$.
Moreover, let~$I_j := (G_j =(V_j, E_j), s_j, t_j, d,a)$.

We describe how to construct an instance~$I^* = (G^*, s^*, t^*, c^*, \omega^*, d^*, a^*)$ of~\WMCI in polynomial time, where~$d^* + a^* + \lp(G^*) + \Delta(G^*) + \td(G^*) \in (d + a + \lp(G) + \Delta(G) + \td(G) + q)^{\Oh(1)}$ such that~$I^*$ is a yes-instance of~\WMCI if and only if~$I_j$ is a yes-instance of~\MCI for at least one~$j \in [1,2^q]$.

We add the following vertices and edges to the graph~$G^*$:

\begin{itemize}
\item We add a copy of the graph~$G_j$ for each~$j\in[1,2^q]$ to~$G^*$.
\item Furthermore, we add~$6q$ new vertices~$Z_j:=\{z^1_j, \ldots , z^{6q}_j\}$ and we add the the edges~$\{s_j,z^i_j\}$ and~$\{z^i_j,t_j\}$ for each~$j\in[1,2^q]$ and each~$i\in[1,6q]$ to~$G^*$.
\item Next, we add a full binary tree~$B$ with height~$q$ to~$G^*$. 
We denote the vertices on level~$\ell$ as~$b_\ell^1, \ldots , b_\ell^{2^\ell}$ for each~$\ell\in[0,q]$.
Hence, vertex~$b_\ell^i$ for some~$\ell\in[0,q-1]$ and some~$i\in[1,2^\ell]$ has the neighbors~$b_{\ell+1}^{2i-1}$ and~$b_{\ell+1}^{2i}$ in the next level.
Now, we identify vertex~$s_j$ with the leaf~$b_q^j$ for each~$j\in[1,2^q]$ and we identify vertex~$s^*$ with the root~$b_0^1$.
\item Analogously, we add a full binary tree~$R$ of height~$q$ with the vertices~$r_\ell^1, \ldots , r_\ell^{2^\ell}$ on level~$\ell\in[0,q]$.
Similarly, we identify vertex~$t^*$ with the root~$r_0^1$ and identify vertex~$t_j$ with the leaf~$r_q^j$ for each~$j\in[1,2^q]$.
\end{itemize}

Next, we set~$d^*:=2q(d+1)+d$ and~$a^*:=7q+a$. 
Afterwards, we set~$\omega^*(e):=1$ for each edge~$e\in E(G^*)$.
We define the costs of each edge in~$E(G^*)$ as follows:

\begin{itemize}
\item For each~$e\in E_j$ for some~$j\in[1,2^q]$ we set~$c(e):=1$.
\item We set~$c(\{s_j,z_j^i\}):=d^*+1$ and~$c(\{z_j^i,t_j\}):=d^*+1$ for each~$j\in[1,2^q]$ and each~$i\in[1,6q]$.
\item For each edge~$e$ in one of the full binary trees we set~$c(e):=d+1$.
\end{itemize}

This completes the construction of~$I^*$.
Now, we prove that the parameters~$a^*$,~$d^*$,~$\Delta(G^*)$, and~$\lp(G^*)$ are bounded by~$(a+d+\Delta(G)+\lp(G)+q)^{\Oh(1)}$. 

\begin{itemize}
\item Since~$d^*=2q(d+1)$ and~$a^*=7q+a$, the statement is clear for~$a^*$ and~$d^*$.
\item Each vertex in~$B\cup R$ except the leaves have degree at most three.
Furthermore, each vertex~$z^i_j$ for some~$j\in[1,2^q]$ and some~$i\in[1,6q]$ as degree two and each vertex in~$V_j\setminus\{s_j, t_j\}$ has degree at most~$\Delta(G)$.
Note that vertex~$s_j$ and~$t_j$ for each~$j\in[1,2^q]$ has degree at most~$6q+1+\Delta(G)$.
Hence, the statement is true for~$\Delta(G^*)$.
\item Observe that the graph obtained from contracting all vertices in~$W_j:=V(G_j)\cup\{z^1_j, \ldots , z^{6q}_j\}$ in the graph~$G^*$ into one vertex for each~$j\in[1,2^q]$ is a mirror fully binary tree of height~$q$.
Observe that by construction,~$\lp(G^*(W_j))\in \Oh(\lp(G))$ for each~$j\in [1,2^q]$.
Thus, by Lemma~\ref{lem-longest-path-in-mirror-fully-binary trees} we obtain that~$\lp(G^*)\in\Oh(\lp(G)\cdot q^2)$.
Hence, the statement is true for~$\lp(G^*)$.
\item Since~$\td(G_j) = \td(G)$ for each~$j\in [1, 2^q]$, the tree-depth of~$G'_j = G^*[V_j \cup Z_j]$ is at most~$\td(G) + 2$. Hence, there is a directed tree~$T_j(V_j \cup Z_j, A_j)$ of depth at most~$\td(G) + 2$ and root~$s_j$, such that for each edge~$\{u,w\}\in E(G'_j)$ either~$u$ is an ancestor of~$w$ in~$T_j$ or vice versa.
We define a directed tree~$T^*=(V(G^*), A^*)$ as follows. 
The tree~$T_j$ is a subtree of~$T^*$ for each~$j\in [1,2^q]$.
The vertex~$b^1_0$ is the root of~$T^*$ and for each~$\ell\in [0,q-1]$ and each~$i\in[1, 2^\ell]$,~$A^*$ contains the arcs~$(b_\ell^i, r_\ell^i),(r_\ell^i, b_{\ell+1}^{2i})$, and~$(r_\ell^i, b_{\ell+1}^{2i + 1})$.
Recall that~$b_q^j = s_j$.
Since~$T_j$ is a subtree of~$T^*$ it follows that for each edge~$\{u,w\}\in E(G^*)$ either~$u$ is an ancestor of~$w$ in~$T^*$ or vice versa.
Moreover, since~$T^*$ has depth at most~$2q + \td(G)$ we obtain the stated bound on the tree-depth of~$G^*$.

\end{itemize}

Next, we prove the correctness.
That is, we show that at least one instance~$I_j$ has a solution~$D_j$ with cost at most~$d$ for some~$j\in[1,2^q]$ if and only if~$I^*$ has a solution~$D$ with cost at most~$d^*$.

$(\Rightarrow)$
Let~$D_j$ be a solution of~$I_j$ with cost at most~$d$.

Let~$P^s_j$ be the unique~$(s^*, s_j)$-path in~$B$ and let~$P^t_j$ be the unique~$(t_j, t^*)$-path in~$R$.
We set~$D:=D_j\cup E(P^s_j)\cup E(P^t_j)$.

First, we show that~$c^*(D)\le d^*$.
Since~$B$ and~$R$ are full binary trees of height~$q$, both paths~$P^s_j$ and~$P^t_j$ consist of exactly~$q$ edges.
Recall that each edge in both~$B$ and~$R$ has cost~$d+1$.
Since~$|D_j|\le d$ and~$c(e)=1$ for each edge~$e\in E_j$, we obtain~$c^*(D)\le 2q(d+1)+d$.

Second, we prove that there is no~$(s^*,t^*)$-cut~$A\subseteq E(G^*)\setminus D$ in~$G^*$ with~$\omega^*(A)\le a^*$. 
To this end, we present~$a^*+1$ many~$(s^*,t^*)$ paths in~$G^*$ whose edge sets may only intersect in~$D$.
Together with~$\omega^*(e)=1$ for each edge~$e\in E(G^*)$ we then conclude that~$\omega^*(M)\ge a^*+1$.
We use the following notation:
For two paths~$P_1=(v_1, \dots, v_k)$ and~$P_2=(w_1, \dots, w_r)$ in~$G^*$ where~$w_r = v_1$, we let~$P_1 \multimap P_2 := (v_1, \dots, v_k = w_1, \dots, w_r)$ denote the \emph{merge} of~$P_1$ and~$P_2$.

\begin{itemize}
\item Let~$P^{s,\ell}_j$ be the subpath of~$P^s_j$ until level~$\ell\in[0,q-1]$.
Since~$B$ is a full binary tree, let~$b^i_{\ell+1}$ be the child of the endpoint of~$P^{s,\ell}_j$ which is not contained in~$P^s_j$.
The subpath~$P^{\ell, t}_j$ and the vertex~$r^i_{\ell+1}$ are defined similarly.
We consider the~$(s^*,t^*)$-path~$P^{s,\ell}_j \cdot (b^i_{\ell+1},b^{2i}_{\ell+2}, \ldots, b^{2^{q-\ell}\cdot i}_q=s_{2^{q-\ell} \cdot i},z^1_{2^{q-\ell}\cdot i},t_{2^{q-\ell}\cdot i}=r^{2^{q-\ell}\cdot i}_q, \ldots, r^{2i}_{\ell+2},r^i_{\ell+1})\cdot P^{\ell, t}_j$.
Since~$\ell\in[0,q-1]$, these are~$q$ paths in total.

\item Observe that~$P^s_j \multimap (s_j, z_j^i, t_j) \multimap  P^t_j$ is an~$(s^*,t^*)$-path for each~$i\in[1,6q]$.
Hence, these are~$6q$ paths in total.
\item Since~$D_j$ is a solution of~$I_j$, there are~$a+1$ many~$(s_j, t_j)$-paths~$P_1, \ldots , P_{a+1}$ in~$G_j$ whose edge set may only intersect in~$D_j$.
Since~$D_j\subseteq D$,~$P^s_j \multimap P_i \multimap P^t_j$ is an~$(s^*,t^*)$-path for each~$i\in[1,a+1]$ such that~$P_i\subseteq E_j$.
Hence, these are~$a+1$ paths in total.
\end{itemize}

Thus,~$G^*$ contains at least~$a+1$ many~$(s^*,t^*)$-paths whose edge set may only intersect in~$D$ and hence~$I^*$ is a yes-instance of \WMCI.

$(\Leftarrow)$
Conversely, let~$D$ be a solution with cost at most~$d^*$ of~$I^*$.
At the beginning, we prove the following statement.

\begin{myclaim}
\label{claim-wmci-no-poly-instance-choice}
For each solution~$D$ of~$I^*$ with cost at most~$d^*$, there exists a~$j\in[1,2^q]$ such that~$E(P^s_j)\subseteq D$ for the unique path~$P^s_j$ from~$s^*$ to~$s_j$ and~$E(P^t_j)\subseteq D$ for the unique path~$P^t_j$ from~$t^*$ to~$t_j$.
\end{myclaim}
\begin{claimproof}
Assume towards a contradiction that this is not the case.
We define~$B_s:=E(B)\cap D$ and~$R_t:=E(R)\cap D$ as the set of protected edges in the binary trees~$B$ and~$R$.
Note that since~$c^*(e)=d+1$ for each edge~$e$ in the binary trees~$B$ and~$R$ and~$d^*=2q(d+1)+d$, we have~$|B_s|+|R_t|\le 2q$.
By~$Z_s$ we denote the connected component of~$G^*[B_s]$ containing vertex~$s^*$.
Since~$|B_s|\le 2q$, we conclude that~$Z_s$ contains at most~$2q+1$ vertices.
Since~$B$ is a binary tree, each vertex in~$B$ has degree at most three.
Recall that only vertices in level~$q$ of~$B$ have neighbors outside of~$B$.
We set~$X:=E_B(Z_s,N(Z_s))$. 
Note that~$|X|\le 3\cdot(2q+1)\le 6q+3$.

First, we consider the case that~$G^*[B_s]$ does not contain the path~$P^s_j$ as an induced subgraph for any~$j\in[1,2^q]$.
We set~$A:=X$ and show that~$M$ is an~$(s^*,t^*)$-cut in~$G^*$. 
Note that~$|M|\le 6q+3\le a^*$.
Observe that~$A$ avoids~$D$ since~$Z_s$ is a connected component in~$G^*[D]$ and~$A$ contains only adjacent edges of~$Z_s$.
Thus,~$A$ is an~$(s^*,t^*)$-cut in~$G^*$ that avoids~$D$ with~$\omega^*(A)\le a^*$, a contradiction.
Analogously, we can prove that~$G^*[R_t]$ is a path~$P^t_{j'}$ for some~$j'\in[1,2^q]$.

Second, we consider the case that~$G^*[B_s]$ is a path~$P^s_j$ for some~$j\in[1,2^q]$ and that~$G^*[R_t]$ is a path~$P^t_{j'}$ for some~$j'\in[1,2^q]$ and~$j\ne j'$.
Recall that~$c(e)=d+1$ for each edge~$e$ in the binary trees~$B$ and~$R$ and that~$d^*=2q(d+1)+d$.
Furthermore, note that~$|B_s|=q=|R_t|$ and that~$Z_s$ contains~$q+1$ vertices.
Thus,~$D$ contains no other edges of~$R$ than~$R_t$.
In particular,~$e^*:=\{r_q^j,r_{q-1}^{\lceil j/2\rceil}\}\notin D$.
Recall that~$t_j=r_q^j$.
We define~$A:=\{e^*\}\cup X$.
Clearly,~$A$ avoids~$D$.
Since~$q\ge 2$, we conclude that~$|A|\le 1+6q+1\le 7q$.
It remains to show that~$A$ is an~$(s^*,t^*)$-cut in~$G^*$.
Since~$X \subseteq A$, every~$(s^*,t^*)$-path~$P^*$ in~$G^*- X$ starts with~$P^s_j$ followed by an~$(s_j, t_j)$-path. 
Moreover,~$P^*$ has to contains the edge~$e^*$.
Since~$e^*\in A$,~$A$ is indeed an~$(s^*,t^*)$-cut in~$G^*$ with~$|A|\le a^*$ that avoids~$D$, a contradiction.
\end{claimproof}

By Claim~\ref{claim-wmci-no-poly-instance-choice}, let~$j\in[1,2^q]$ such that~$E(P^s_j)\cup E(P^t_j)\subseteq D$.
We define~$D_j:=D\cap E_j$.
Observe that since~$c(e)=d+1$ for each edge in both binary trees~$B$ and~$R$, each edge incident with vertex~$z^i_{j'}$ for some~$j'\in[1,2^q]$ and some~$i\in[1,6q]$ has cost~$d^*+1$, and each edge in the copy of~$G_{j'}$ for some~$j'\in[1,2^q]$ has costs one, we conclude that~$|D_j|\le d$.
In the following, we show that~$D_j$ is a solution of~$I_j$.

Assume towards a contradiction that there is an~$(s_j,t_j)$-cut~$A_j\subseteq E_j$ of size at most~$a$ in~$G_j$ that avoids~$D_j$.
We set~$A := A_j \cup \{\{s_j, z^i_j\}\mid i \in [1,6q]\} \cup X$, where~$X := E_B(V(P^s_j), N(V(P^s_j)))$.
Note that~$A$ avoids~$D$.
By the fact that~$B$ is a binary tree of depth~$q$, it follows that~$|M| \leq |M_j| + 6q + q \leq a + 7q = a^*$.

Since~$X \subseteq A$, every~$(s^*,t^*)$-path~$P^*$ in~$G^*- X$ starts with~$P^s_j$ followed by an~$(s_j, t_j)$-path. 
Moreover, since~$A_j$ is an~$(s_j, t_j)$-cut in~$G_j$ and~$\{\{s_j, z^i_j\}\mid i \in [1,6q]\} \subseteq A$,~$A$ is an~$(s^*, t^*)$-cut of capacity at most~$a^*$ in~$G^*$, a contradiction.
Hence,~$D_j$ is a solution with cost at most~$d$ of~$I_j$ and, thus,~$I_j$ is a yes-instance of~\WMCI.
\lncsqed\end{proof}
\fi

\bibliographystyle{plainurl}
\bibliography{bib}
\end{document}